\documentclass[journal]{IEEEtran}%
\usepackage{graphicx}
\usepackage{epsfig}
\usepackage{times}
\usepackage{amsmath}
\usepackage{thmtools,thm-restate}
\usepackage{amssymb}
\usepackage[section]{placeins}
\usepackage{placeins}
\usepackage{amsmath}
\usepackage{multirow}
\usepackage{graphicx}
\usepackage[normalem]{ulem}
\usepackage{subcaption}
\usepackage{algorithm,tabularx}
\usepackage{algpseudocode}
\usepackage{amsfonts}
\usepackage{cases}
\usepackage{romannum}
\usepackage{textcomp}
\usepackage{xcolor}%
\providecommand{\U}[1]{\protect\rule{.1in}{.1in}}
\providecommand{\U}[1]{\protect\rule{.1in}{.1in}}
\IEEEoverridecommandlockouts
\usepackage{amsthm}
\newcommand{\R}{\mathbb{R}}
\newcommand{\Z}{\mathbb{Z}}

\newtheorem{theorem}{Theorem}
\newtheorem{corollary}{Corollary}
\newtheorem{lemma}{Lemma}

\newtheorem{remark}{Remark}
\newtheorem{definition}{Definition}

\useunder{\uline}{\ul}{}
\makeatletter
\newcommand{\multiline}[1]{  \begin{tabularx}{\dimexpr\linewidth-\ALG@thistlm}[t]{@{}X@{}}
#1
\end{tabularx}
}
\makeatother
\usepackage{cite}
\usepackage{hyperref}
\usepackage{picins}
\usepackage{centernot}
\usepackage{enumitem}
\setlist[itemize]{leftmargin=*}
\usepackage{makecell}
\usepackage[normalem]{ulem}

\begin{document}


\title{Smooth Robustness Measures for Symbolic Control Via Signal Temporal Logic}

\author{Shirantha Welikala, Hai Lin and Panos J. Antsaklis 
\thanks{The authors gratefully acknowledge the fruitful discussions with Vince Kurtz about different smooth robustness measures and their implementation.}
\thanks{The support of the National Science Foundation (Grant No. IIS-1724070, CNS-1830335, IIS-2007949) is gratefully acknowledged.}
\thanks{The authors are with the Department of Electrical Engineering, College of Engineering, University of Notre Dame, IN 46556, \texttt{{\small \{wwelikal,hlin1,pantsakl\}@bu.edu}}.}}
\maketitle

\begin{abstract}
Symbolic control problems aim to synthesize control policies for dynamical systems under complex temporal specifications. For such problems, Signal Temporal Logic (STL) is increasingly used as the formal specification language due to its rich expressiveness. 
Moreover, the degree of satisfaction of STL specifications can be evaluated using \emph{STL robust semantics} as a scalar robustness measure. This capability of STL enables transforming a symbolic control problem into an optimization problem that optimizes the corresponding robustness measure. 
However, since these robustness measures are non-smooth and non-convex, exact solutions can only be computed using computationally inefficient mixed-integer programming techniques that do not scale well. Therefore, recent literature has focused on using smooth approximations of these robustness measures to apply scalable and computationally efficient gradient-based methods to find local optima solutions. 
In this paper, we first generalize two recently established smooth robustness measures (SRMs) and two new ones and discuss their strengths and weaknesses. 
Next, we propose \emph{STL error semantics} to characterize the approximation errors associated with different SRMs under different parameter configurations. This allows one to sensibly select an SRM (to optimize) along with its parameter values. 
We then propose \emph{STL gradient semantics} to derive explicit gradients of SRMs leading to improve computational efficiency as well as accuracy compared to when using numerically estimated gradients.
Finally, these contributions are highlighted using extensive simulation results.
\end{abstract}

\thispagestyle{empty} \pagestyle{empty}


\section{Introduction}

A symbolic control problem arises when a dynamical system needs to be controlled such that its trajectory satisfies some complex temporal specifications. For example, such specifications may include requirements expected from the dynamical system, like achieving time-sensitive goals, avoiding undesirable behaviors and ensuring desirable behaviors over certain periods of time. To precisely express such temporal specifications, temporal logics like Linear Temporal Logic \cite{Pnueli1977}, Computation Tree Logic \cite{Clarke1982}, Interval Temporal Logic \cite{Schwartz1983},  Metric Temporal logic \cite{Koymans1990}, Time Window Temporal Logic \cite{Vasile2017} and Signal Temporal Logic (STL) \cite{Maler2004} can be used.

Among these temporal logics, STL has the advantage of having \emph{STL robust semantics} that can systematically transform the given temporal specifications and a system trajectory (i.e., a ``signal'') into a scalar robustness measure - which is indicative of both the satisfaction/violation and the degree of satisfaction/violation of the specification via its sign and magnitude, respectively. This unique property enables STL to transform a symbolic control problem into an optimization problem that finds the control sequence (signal) that maximizes the corresponding robustness measure. Due to these capabilities of STL, in recent years, it has been increasingly used in control systems related application domains such as in planning \cite{Mehdipour2019}, robotics \cite{Haghighi2019}, multi-agent systems \cite{Pant2018}, automotive systems \cite{Fainekos2012}, cyber-physical systems \cite{Raman2015} and biological systems \cite{Sankaranarayanan2012,Mehdipour2019b}.

The conventional robustness measure is non-smooth and non-convex due to the involved min and max operators in its definition. Therefore, optimizing such a robustness measure is an extremely challenging task. The work in \cite{Belta2019} formulates this optimization problem as a Mixed-Integer Program (MIP) so as to obtain an exact solution. However, this MIP formulation scales poorly with respect to many aspects of the original optimization problem (e.g., the complexity of the specification). Moreover, MIPs are inherently computationally inefficient to solve due to their complexity properties. With regard to this optimization problem, prior work have also explored alternative MIP formulations \cite{Saha2016,Raman2015} as well as heuristic \cite{Abbas2012,Lavalle2001} and non-smooth \cite{Abbas2013,Abbas2014} optimization approaches, but with no solution to the aforementioned scalability and computational efficiency concerns.

To address these concerns, several recent works such as \cite{Mehdipour2019,Haghighi2019,Gilpin2021} pioneered by \cite{Pant2017} propose to use smooth approximations of the conventional robustness measure so as to exploit the scalable, computationally efficient and well-established gradient-based optimization methods. It is worth noting that even though gradient-based optimization methods are susceptible to the issue of local optima, there are several efficient and practical ways to overcome this issue \cite{Lasdon2008,Welikala2019J1,Hoos2005}. Moreover, the consensus in the literature \cite{Pant2017,Pant2018,Mehdipour2019,Haghighi2019,Gilpin2021} seems to indicate that the scalability and the computational efficiency offered by adopting gradient-based optimization methods outweigh the effect of any such drawbacks.

Typically, an approximate smooth robustness measure (SRM) is obtained by replacing (approximating) the min and max operators in the conventional robustness measure with some smooth operators. Therefore, depending on the nature of the used smooth operators, the resulting SRM will have different characteristics. For example, \cite{Pant2017} proposes to use quasi-min and quasi-max \cite{Lange2014} operators as the smooth operators. However, even though the resulting SRM is globally smooth, it is not sound, i.e., achieving a positive value for the SRM does not imply the satisfaction of the specification \cite{Pant2017}. To recover this soundness property, \cite{Mehdipour2019} proposes to use smooth operators based on arithmetic and geometric means (partly motivated by \cite{Haghighi2019}). Even though the resulting SRM ensures soundness and also leads to more conservative results (i.e., more uniformly robust to external disturbances), it is not globally smooth \cite{Mehdipour2019}. Motivated by this predicament, the subsequent work \cite{Gilpin2021} (which also is the predecessor to this work) proposes to use quasi-min and soft-max \cite{Lange2014} operators as smooth operators. The resulting SRM is sound, globally smooth, and also offers a tunable level of conservativeness \cite{Gilpin2021} (via selecting the parameters of the used smooth operators).

Despite these successive recent improvements, there still are many contributions to be made in this research. For example, even though the proposed SRMs in \cite{Mehdipour2019} and \cite{Gilpin2021} are sound, they are not reverse sound, i.e., achieving a negative value for the SRM does not imply the violation of the specification. In fact, in some applications (e.g., in falsification problems \cite{Abbas2012}), the reverse soundness property is far more crucial than the soundness property. To address this need, inspired by \cite{Pant2017,Gilpin2021}, in this paper, we propose a new SRM that is reverse sound and also globally smooth. Along the same lines, we propose another SRM that aims to provide more flexibility when tuning the level of conservativeness (still, via selecting the parameters of the used smooth operators).

A natural question to ask at this point is: ``How one should select these smooth operator parameters?'' Typically, smaller parameter values lead to more conservative results, while larger parameter values lead to more accurate SRM (i.e., closer to the actual robustness measure). Authors in \cite{Pant2017} motivate the latter accuracy argument as it leads to ensure properties like soundness and reverse soundness (together called completeness) asymptotically. However, authors in \cite{Gilpin2021} point out the prior conservativeness argument. Nevertheless, both works do not provide a systematic approach to select these parameters. To answer this question, in this paper, we propose \emph{STL error semantics} to determine a bound for the approximation error associated with the used SRM in terms of the used operator parameters. Note that, since the width of such an error bound is a measure of the accuracy of the used SRM in terms of the used operator parameters, it can easily be used to tune the operator parameters and compare different (existing and newly proposed) SRMs. 

As mentioned before, the primary motivation behind using an SRM (as opposed to using the non-smooth conventional robustness measure) is to enable the use of efficient and scalable gradient-based optimization methods \cite{Pant2017,Pant2018,Haghighi2019,Mehdipour2019,Gilpin2021}. Therefore, intuitively, to best reap these benefits of gradient-based optimization, having an accurate as well as efficient way to evaluate the required ``gradients'' is of prime importance (here, the ``gradients'' refer to the gradients of the used SRM with respect to the system controls).  Clearly, the use of numerically evaluated gradients is both computationally inefficient as well as inaccurate \cite{Pant2017}. However, to the best of the authors’ knowledge, experimental results reported in existing literature \cite{Pant2017,Pant2018,Haghighi2019,Mehdipour2019,Gilpin2021} use numerically evaluated gradients defined via finite differences \cite{Pant2017} or software packages like AutoGrad \cite{Maclaurin2015}. Motivated by this, we propose \emph{STL gradient semantics} to determine the explicit gradient of any used SRM. Such derived explicit gradients are accurate (by definition) and significantly efficient compared to numerically evaluated gradients, and thus, allows us to reap the benefits of gradient-based optimization methods to the fullest.

In all, our contributions can be summarized as follows:
\begin{enumerate}
    \item Two new SRMs are proposed. One is guaranteed to be reverse sound, while the other is intended to provide more flexibility when tuning the level of conservativeness.
    \item We propose \emph{STL error semantics} to determine a bound for the approximation error of the used SRM in terms of the used operator parameters, enabling one to sensibly select an SRM along with its smooth parameter values.
    \item We propose \emph{STL gradient semantics} to derive the explicit gradient of the used SRM, leading to improve the accuracy and the efficiency (of both gradients and results) compared to when using numerically estimated gradients. 
    \item An extensive collection of simulation results has been reported along with source codes for reproduction and reuse purposes. 
\end{enumerate}


This paper is organized as follows. 
Section \ref{Sec:ProblemFormulation} introduces the considered class of symbolic control problems and the \emph{STL robust semantics}. Section \ref{Sec:SmoothRobustnessMeasures} provides the details of existing and new SRMs and introduces the \emph{smooth STL robust semantics}. The proposed \emph{STL error semantics} that can be used to characterize the approximation errors associated with different SRMs are discussed in Section \ref{Sec:ApproximationErrors}. The proposed \emph{STL gradient semantics} that can be used to derive explicit gradients of SRMs are discussed in Section \ref{Sec:Gradients}. Several observed simulation results, along with few interesting future research directions, are provided in Section \ref{Sec:SimulationResults} before concluding the paper in Section \ref{Sec:Conclusion}.

\section{Problem Formulation}
\label{Sec:ProblemFormulation}

We consider a non-linear discrete-time system of the form:
\begin{equation}\label{Eq:AgentDynamics}
\begin{aligned} 
    x_{t+1} =& f(x_t,u_t),\\
        y_t =& g(x_t,u_t),
\end{aligned}
\end{equation}
where $x_t\in \mathcal{X} \subseteq \R^n$ is the system state, $u_t\in\mathcal{U} \subseteq \R^m$ is the control input, and $y_t\in\mathcal{Y} \subseteq \R^p$ is the output signal. 
It is also assumed that $f(\cdot,\cdot)$ and $g(\cdot,\cdot)$ in \eqref{Eq:AgentDynamics} are continuous and differentiable. Given a finite horizon $T$, an initial condition $x_0$ and a bounded-time STL specification $\varphi_0$, our goal is to synthesize a control input sequence $u \triangleq \{u_t\}_{t\in[0,T]}$ such that the resulting output signal $y \triangleq \{y_t\}_{t\in[0,T]}$ satisfies $\varphi_0$ (here, $[0,T]$ denotes the discrete interval $\{0,1,2,\ldots,T\}$).

\subsection{STL Robust Semantics}
STL was first introduced in \cite{Maler2004} as a variant of temporal logic tailored for specifying desired temporal properties of real-valued signals. In particular, an STL specification $\varphi$ is built recursively from the predicates using the STL syntax: 
\begin{equation}\label{Eq:STLSyntax}
    \varphi :=  \pi \vert \neg \varphi \vert \varphi_1 \land \varphi_2 \vert \varphi_1 \mathbf{U}_{[t_1,t_2]}\varphi_2,
\end{equation}
where $\pi = (\mu^\pi(y_t)\geq 0)$ is a predicate defined by the function $\mu^\pi: \R^p\rightarrow \R$. While some STL synthesis methods require $\mu^\pi(\cdot)$ to be linear, in our approach, we only require $\mu^\pi(\cdot)$ to be differentiable. The \emph{conjunction} ($\land$) and \emph{negation} ($\neg$) operators can be used to derive other Boolean operators such as \emph{disjunction} ($\lor$) and \emph{implication} ($\implies$). Similarly, the \emph{until} ($\varphi_1 \mathbf{U}_{[t_1,t_2]} \varphi_2$) operator can be used to derive other temporal operators such as \emph{eventually} ($\mathbf{F}_{[t_1,t_2]}\varphi$) and \emph{always}  ($\mathbf{G}_{[t_1,t_2]}\varphi$). Here $[t_1,t_2]$ represents a discrete interval starting from a time point $t_1$ to a time point $t_2 \geq t_1$  where $t_1,t_2 \in \Z$. Since we restrict ourselves to bounded-time STL specifications, $t_2<\infty$.   

For example, the STL specification $\varphi = (\textbf{G}_{[a,b]}\varphi_1) \land (\textbf{F}_{[b,c]}\varphi_2)$ states that ``$\varphi_1$ must be True at all time points in $[a,b]$ and $\varphi_2$ must be True at some time point in $[b,c]$.'' Here, each $\varphi_i, i\in\{1,2\}$ is an STL specification itself, which incidentally can be a predicate $\varphi_i = \pi_i = (\mu^{\pi_i}(y_t)\geq 0)$. In the latter case, notice the dependence of $\varphi_i$ on the time point $t \in [0,T] \supseteq [a,c]$.

We use the \emph{STL robust semantics} (defined below) to assign a scalar value (called the \emph{robustness measure}) for a signal with respect to a given specification. This particular robustness measure is positive if and only if the signal satisfies the given specification.
In other words, $\rho^\varphi(y)>0 \iff y\vDash \varphi$, where $\rho^\varphi(y)$ is the robustness measure of the signal $y$ with respect to the specification $\varphi$. Let us use the notation $(y,t)\triangleq\{y_\tau\}_{\tau\in[t,T]}$ to denote the suffix of the signal $y (\triangleq\{y_\tau\}_{\tau\in[0,T]})$ starting from time point $t$ (note that, $y\equiv(y,0)$). 

\begin{definition}\label{Def:STLRobustSemantics}
(STL Robust Semantics)
\begin{itemize}
    \item $y \vDash \varphi \iff \rho^\varphi((y,0)) > 0$
    \item $y \nvDash \varphi \iff \rho^\varphi((y,0)) < 0$
    \item $\rho^\pi((y,t)) = \mu^\pi(y_t)$
    \item $\rho^{\neg\varphi}((y,t)) = -\rho^\varphi((y,t))$
    \item $\rho^{\varphi_1 \land \varphi_2}((y,t)) = \min \{\rho^{\varphi_1}((y,t)),\, \rho^{\varphi_2}((y,t))\}$
    \item $\rho^{\varphi_1 \lor \varphi_2}((y,t)) = \max \{\rho^{\varphi_1}((y,t)),\, \rho^{\varphi_2}((y,t))\}$
    \item $\rho^{\mathbf{F}_{[t_1,t_2]}\varphi}((y,t)) = \max_{\tau\in[t+t_1,t+t_2]}\{\rho^\varphi((y,\tau))\}$
    \item $\rho^{\mathbf{G}_{[t_1,t_2]}\varphi}((y,t)) = \min_{\tau\in[t+t_1,t+t_2]}\{\rho^\varphi((y,\tau))\}$
    \item $\rho^{\varphi_1\mathbf{U}_{[t_1,t_2]}\varphi_2}((y,t)) = \max_{\tau\in[t+t_1,t+t_2]}\{\min\{$\newline
    $\rho^{\varphi_1}((y,\tau)),\min_{\delta\in[t+t_1,\tau]}\{\rho^{\varphi_2}((y,\delta))\}\}$
\end{itemize}
\end{definition}


\subsection{Synthesis Problem as an Optimization Problem}

Finding a control input sequence $u$ for the system \eqref{Eq:AgentDynamics} so that the resulting output signal $y$ satisfies the given specification $\varphi_0$ is formally known as the \emph{synthesis problem}. Using STL robust semantics discussed above, this synthesis problem can be stated as an optimization problem of the form:
\begin{equation}\label{Eq:SynthesisProblem}
    \begin{aligned}
        u^* = &\ \underset{u}{\arg\max}& \ &\rho^{\varphi_0}((y,0))\\
        &\ \mbox{Subject to:}&\ &x_{t+1} = f(x_t,u_t),\ x_0 \mbox{ is a given,}\\
        & & \ &y_t = g(x_t,u_t),\\
        & & &y_t \in \mathcal{Y},\ x_t\in\mathcal{X},\ u_t\in\mathcal{U},\ \forall t\in[0,T].
    \end{aligned}
\end{equation}

Note that the three feasibility constraints in the last line of \eqref{Eq:SynthesisProblem} can be included in the robustness measure (i.e., in the objective function of \eqref{Eq:SynthesisProblem}) by considering them as STL specifications, say $\varphi_y, \varphi_x, \varphi_u$, on respective signals $y, x, u$ (similar to $y$ and $u$, $x \triangleq \{x_t\}_{t\in [0,T]}\equiv(x,0)$). For example, if $\mathcal{U}=\{u: u \in \R^m, \Vert u \Vert \leq a\}$, then, $\varphi_u = \textbf{G}_{[0,T]}(a - \Vert u_t \Vert \geq 0)$.

Therefore, \eqref{Eq:SynthesisProblem} can be re-stated as 
\begin{equation}\label{Eq:SynthesisProblem2}
    \begin{aligned}
        u^* = &\ \underset{u}{\arg\max}& \ &\rho^{\varphi}(y,x,u)\\
        &\ \mbox{Subject to:}&\ &x_{t+1} = f(x_t,u_t),\ x_0 \mbox{ is a given,}\\
        & & \ &y_t = g(x_t,u_t),\ \forall t\in[0,T],
    \end{aligned}
\end{equation}
where $\varphi \triangleq \varphi_0 \land \varphi_y \land \varphi_x \land \varphi_u$, and hence, $\rho^{\varphi}(y,x,u) = \max\{\rho^{\varphi_0}(y),\, \rho^{\varphi_y}(y),\, \rho^{\varphi_x}(x),\, \rho^{\varphi_u}(u)\}$. With a slight abuse of notation, we can replace the STL specification $\varphi_0 \land \varphi_y$ simply by $\varphi_y$. As a result, the robustness measure objective in \eqref{Eq:SynthesisProblem2} can be written of as 
\begin{equation}\label{Eq:SynthesisProblem2Objective}
    \rho^{\varphi}(y,x,u) \triangleq \max\{\rho^{\varphi_y}(y),\, \rho^{\varphi_x}(x),\, \rho^{\varphi_u}(u)\}.
\end{equation}

It is worth pointing out that the state signal $x$ and the output signal $y$ are both fully determined by the control signal $u$ through the agent dynamics \eqref{Eq:AgentDynamics} and the given initial condition $x_0$. Exploiting this dependence, we can define a composite control dependent signal $s(u) \triangleq \{s_t(u)\}_{t\in[0,T]}$ where $s_t(u) \triangleq [y_t^\top(u),x_t^\top(u),u_t^\top]^\top \in \R^{p+n+m}$. Hence, $\rho^{\varphi}(y,x,u)$ in \eqref{Eq:SynthesisProblem2Objective} can be written as $\rho^{\varphi}(s(u)) = \rho^{\varphi}(y(u),x(u),u)$. Therefore, \eqref{Eq:SynthesisProblem2} can be re-stated as an unconstrained optimization problem:
\begin{equation}\label{Eq:SynthesisProblem3}
u^* = \underset{u}{\arg\max} \ \rho^{\varphi}(s(u)),    
\end{equation}
where the objective function $\rho^{\varphi}(s(u))$ is simply the robustness measure of the signal $s(u)$ with respect to the specification $\varphi$.

According to the STL robust semantics, if the determined optimal control signal $u^*$ in \eqref{Eq:SynthesisProblem3} is such that $\rho^{\varphi}(s(u^*))>0$, the signal $s(u^*)$ is guaranteed to satisfy the specification $\varphi$. In other words, denoting the corresponding output and state signals respectively as $y^*$ and $x^*$, $y^*$ is guaranteed to satisfy the given specification $\varphi_0$ in \eqref{Eq:SynthesisProblem}, and, signals $y^*$, $x^*$ and $u^*$ will always remain within their respective feasible regions stated in \eqref{Eq:SynthesisProblem}. 


\begin{remark}
In this paper, we focus on solving the optimization problem \eqref{Eq:SynthesisProblem3} (same as \eqref{Eq:SynthesisProblem} and \eqref{Eq:SynthesisProblem2}) to determine the optimal controls for the synthesis problem. An alternative energy-aware approach would be to solve an optimization problem of the form
\begin{equation}\label{Eq:SynthesisProblem4}
u^* = \underset{u}{\arg\max} \ \rho^{\varphi}(s(u)) - \alpha \Vert u \Vert^2, \end{equation}
where $\alpha\in\R$ is a scaling factor. Similarly, an alternative receding horizon control approach (preferred if $T$ is large) can also be facilitated - by slightly modifying the form of \eqref{Eq:SynthesisProblem4}.  
\end{remark}

\section{Smooth Robustness Measures (SRMs)}
\label{Sec:SmoothRobustnessMeasures}

The robustness measure objective $\rho^{\varphi}(s(u))$ of the optimization problem \eqref{Eq:SynthesisProblem3} (i.e., the interested synthesis problem) is non-smooth and non-convex due to the involved non-smooth $\min\{\cdot\}$ and $\max\{\cdot\}$ operators in the STL robust semantics. In the special case where the system \eqref{Eq:AgentDynamics} is linear and all the predicates required to define the STL specification $\varphi$ are linear, the problem \eqref{Eq:SynthesisProblem3} can be encoded as a Mixed Integer Program (MIP) \cite{Belta2019}. However, the complexity of this MIP's solution process grows exponentially with the number of predicates and $T$. Therefore, even for such a special (linear) case, solving \eqref{Eq:SynthesisProblem3} is extremely challenging.

To address this challenge, the recent literature \cite{Pant2017,Haghighi2019,Mehdipour2019,Gilpin2021} has focused on replacing the non-smooth robustness measure objective $\rho^{\varphi}(s(u))$ in \eqref{Eq:SynthesisProblem3} with a smooth approximation of it. Such a smooth approximation is denoted as $\tilde{\rho}^{\varphi}(s(u))$ and is referred to as a \emph{smooth robustness measure} (SRM). This modification enables the use of computationally efficient and scalable gradient-based optimization techniques to solve \eqref{Eq:SynthesisProblem3}.  

Typically, an SRM $\tilde{\rho}^{\varphi}(s(u))$ is obtained by replacing (approximating) the non-smooth $\min\{\cdot\}$ and  $\max\{\cdot\}$ operators involved in $\rho^{\varphi}(s(u))$ by respective smooth-min and smooth-max operators (denoted as $\widetilde{\min}\{\cdot\}$ and $\widetilde{\max}\{\cdot\}$). In particular, to evaluate such an SRM $\tilde{\rho}^{\varphi}(s(u))$, we use the \emph{smooth STL robust semantics} defined below.  

\begin{definition}\label{Def:SmoothSTLRobustSemantics}
(Smooth STL Robust Semantics)
\begin{itemize}
    \item $\tilde{\rho}^\pi((s,t)) = \mu^\pi(s_t)$
    \item $\tilde{\rho}^{\neg\pi}((s,t)) = -\tilde{\rho}^\pi((s,t))$
    \item $\tilde{\rho}^{\varphi_1 \land \varphi_2}((s,t)) = \widetilde{\min} \{\tilde{\rho}^{\varphi_1}((s,t)),\, \tilde{\rho}^{\varphi_2}((s,t))\}$
    \item $\tilde{\rho}^{\varphi_1 \lor \varphi_2}((s,t)) = \widetilde{\max} \{\tilde{\rho}^{\varphi_1}((s,t)),\, \tilde{\rho}^{\varphi_2}((s,t))\}$
    \item $\tilde{\rho}^{\mathbf{F}_{[t_1,t_2]}\varphi}((s,t)) = \widetilde{\max}_{\tau\in[t+t_1,t+t_2]}\{\tilde{\rho}^\varphi((s,\tau))\}$
    \item $\tilde{\rho}^{\mathbf{G}_{[t_1,t_2]}\varphi}((s,t)) = \widetilde{\min}_{\tau\in[t+t_1,t+t_2]}\{\tilde{\rho}^\varphi((s,\tau))\}$
    \item $\tilde{\rho}^{\varphi_1\mathbf{U}_{[t_1,t_2]}\varphi_2}((s,t)) = \widetilde{\max}_{\tau\in[t+t_1,t+t_2]}\{\widetilde{\min}\{
    \newline 
        \tilde{\rho}^{\varphi_1}((s,\tau)),\, \widetilde{\min}_{\delta\in[t+t_1,\tau]}\{\tilde{\rho}^{\varphi_2}((s,\delta))\}\}\}$
\end{itemize}
\end{definition}

In writing the above semantics, for notational simplicity, we have omitted representing the control ($u$) dependence of the composite signal ($s(u)$). Moreover, without loss of generality, we have assumed that the STL specification is given in disjunctive normal form, i.e., the negation operator is only applied to predicates $\pi$. It is worth noting that any STL specification can be written in disjunctive normal form \cite{Gilpin2021}.

\begin{table*}[!t]
\caption{A summary of available different smooth-min ($\widetilde{\min}\{\cdot\}$) and smooth-max ($\widetilde{\max}\{\cdot\}$) operators and their characteristics. Here, $a \triangleq \{a_1,a_2,\ldots,a_m\}$ is a finite set of real numbers (assume to be in descending order) and parameters $k_1,k_2 \in \R_{> 0}$.} 
\label{Tab:SmoothOperators}
\centering
\resizebox{\textwidth}{!}{
\begin{tabular}{|c|l|l|l|}
\hline
\multicolumn{1}{|c|}{Smooth Operator} & 
\multicolumn{1}{c|}{\begin{tabular}[c]{@{}c@{}}Operator Definition\end{tabular}} & 
\multicolumn{1}{c|}{\begin{tabular}[c]{@{}c@{}}Gradient\\ (e.g., $\frac{\partial}{\partial a_i}\widetilde{\min}\{a\}=\,?$) \end{tabular}} & 
\multicolumn{1}{c|}{\begin{tabular}[c]{@{}c@{}}Approximation Error Band\\ 
(e.g., $(\min\{a\}-\widetilde{\min}\{a\})\in\,?$)\end{tabular}} \\ \hline
$k_1$-Quasi-$\min$ & 
$\widetilde{\min}\{a\} \triangleq -\frac{1}{k_1}\log\left(\sum_{i=1}^m e^{-k_1a_i}\right)$ & 
$\frac{e^{-k_1a_i}}{\sum_{j=1}^m e^{-k_1a_j} }$ & 
$[0,\ \frac{\log(m)}{k_1}]$ \\ \hline
$k_2$-Quasi-$\max$ & 
$\widetilde{\max}\{a\} \triangleq \frac{1}{k_2}\log\left(\sum_{i=1}^m e^{k_2a_i}\right)$ & 
$\frac{e^{k_2a_i}}{\sum_{j=1}^m e^{k_2a_j}}$ & 
$[-\frac{\log(m)}{k_2},\ 0]$ \\ \hline
$k_1$-Soft-$\min$ &
$\widetilde{\min}\{a\} \triangleq \frac{\sum_{i=1}^m a_i e^{-k_1a_i}}{\sum_{i=1}^m e^{-k_1a_i}}$ & 
$\frac{e^{-k_1a_i}}{\sum_{j=1}^m e^{-k_1a_j}}\left[1 - k_1(a_i-\widetilde{\min}\{\bar{a}\})\right]$ & $[-\frac{a_1-a_m}{1+\frac{e^{k_1(a_{m-1}-a_m)}}{m-1}},\ 0]$ \\ \hline
$k_2$-Soft-$\max$ &
$\widetilde{\max}\{a\} \triangleq \frac{\sum_{i=1}^m a_i e^{k_2a_i}}{\sum_{i=1}^m e^{k_2a_i}}$ & 
$\frac{e^{k_2a_i}}{\sum_{j=1}^m e^{k_2a_j}}\left[1 - k_2(\widetilde{\max}\{\bar{a}\}-a_i)\right]$ & $[0,\ \frac{a_1-a_m}{1+\frac{e^{k_2(a_1-a_2)}}{m-1}}]$ \\ \hline
\end{tabular}
}
\end{table*}

\begin{table}[t!]
\caption{Different smooth robustness measures (SRMs) defined by the used smooth-min and smooth-max operators.}
\label{Tab:SmoothRobustnessMeasures}
\resizebox{\columnwidth}{!}{
\begin{tabular}{|c|r|c|c|}
\hline
\multicolumn{2}{|c|}{\multirow{2}{*}{\begin{tabular}[c]{@{}c@{}}SRM Type\end{tabular}}} &
  \multicolumn{2}{c|}{\begin{tabular}[c]{@{}c@{}}Smooth-max \\ Operator\end{tabular}} \\ \cline{3-4} 
\multicolumn{2}{|c|}{} &
  $k_2$-Quasi-$\max$ &
  $k_2$-Soft-$\max$ \\ \hline
\multirow{2}{*}{\begin{tabular}[c]{@{}c@{}}Smooth-min \\ Operator\end{tabular}} &
  $k_1$-Quasi-$\min$ &
  \begin{tabular}[c]{@{}c@{}}SRM1 \\ (Proposed in \cite{Pant2017}) \end{tabular} &
  \begin{tabular}[c]{@{}c@{}}SRM2 \\ (Proposed in \cite{Gilpin2021}) \end{tabular} \\ \cline{2-4} 
 &
  $k_1$-Soft-$\min$ &
  \begin{tabular}[c]{@{}c@{}}SRM3 \\ (New)\end{tabular} &
  \begin{tabular}[c]{@{}c@{}}SRM4 \\ (New)\end{tabular} \\ \hline
\end{tabular}
}
\end{table}

Table \ref{Tab:SmoothOperators} summarizes possible two smooth-min operators and two smooth-max operators that can be used in Def. \ref{Def:SmoothSTLRobustSemantics}. Based on the configuration of the used smooth-min and smooth-max operators, as outlined Tab. \ref{Tab:SmoothRobustnessMeasures}, we can define four different SRMs (labeled SRM1-4). In the following four subsections, we discuss each of these different SRMs.

\paragraph*{Notation} Unless stated otherwise, $a \triangleq \{a_1,a_2,\ldots,a_m\}$ is used to represent an arbitrary finite set of real numbers (without loss of generality, $a$ is also assumed to be in descending order). Further, $k_1,k_2\in \R_{>0}$ are used to represent two user defined parameters. We also point out that, we have slightly abused the notation $[\cdot,\cdot]$ by using it to represent closed intervals in both $\R$ and $\Z$. However, we believe its meaning would be clear from the context.



\subsection{SRM1: Proposed in \cite{Pant2017}}

The work in \cite{Pant2017} has proposed to use $k_1$-Quasi-$\min$ and $k_2$-Quasi-$\max$ operators (see Tab. \ref{Tab:SmoothOperators} and \cite{Lange2014}):
\begin{equation}\label{Eq:SmoothMinMaxOperatorsPant}
\begin{aligned}
    \widetilde{\min}\{a\} \triangleq& -\frac{1}{k_1}\log\left(\sum_{i=1}^m e^{-k_1a_i}\right) \mbox{ and }\\
    \widetilde{\max}\{a\} \triangleq& \ \ \ \  \frac{1}{k_2}\log\left(\sum_{i=1}^m e^{k_2a_i}\right),
\end{aligned}
\end{equation}
as smooth-min and smooth-max operators, respectively. The operators in \eqref{Eq:SmoothMinMaxOperatorsPant} are smooth everywhere, and thus, their respective gradients can be derived easily (provided in Tab. \ref{Tab:SmoothOperators}). The following lemma establishes bounds on the approximation errors associated with the operators in \eqref{Eq:SmoothMinMaxOperatorsPant}.

\begin{lemma}\label{Lm:SmoothMinMaxOperatorErrorBoundsPant}
The smooth-min and smooth-max operators in \eqref{Eq:SmoothMinMaxOperatorsPant} satisfy the following respective approximation error bounds:
\begin{equation}\label{Eq:Lm:SmoothMinMaxOperatorErrorBoundsPant}
    \begin{aligned}
        (\min\{a\} - \widetilde{\min}\{a\})
        \in&[0,\,\frac{1}{k_1}\log(1+(m-1)e^{-k_1(a_{m-1}-a_m)})]\\
        \subseteq&[0,\ \frac{\log(m)}{k_1}] \mbox{ and }\\
        (\max\{a\} - \widetilde{\max}\{a\})
        \in&[-\frac{1}{k_2}\log(1+(m-1)e^{-k_2(a_1-a_2)}),\,0]\\
        \subseteq&[-\frac{\log(m)}{k_2},\ 0].
    \end{aligned}
\end{equation}
\end{lemma}
\begin{proof}
Recall that the set $a=\{a_1,a_2,\ldots,a_m\}$ is in descending order. Therefore, the approximation error associated with the smooth-min operator in \eqref{Eq:SmoothMinMaxOperatorsPant} is 
\begin{align}
    (\min\{a\} -& \widetilde{\min}\{a\}) 
    =\ a_m + \frac{1}{k_1}\log\left(\sum_{i=1}^m e^{-k_1a_i}\right) \nonumber \\
    =& \frac{1}{k_1}\log\left(\frac{\sum_{i=1}^m e^{-k_1a_i}}{e^{-k_1a_m}}\right) \nonumber \\
    =& \frac{1}{k_1}\log\left(1+\frac{\sum_{i=1}^{m-1} e^{-k_1a_i}}{e^{-k_1a_m}}\right) \nonumber \\
    \leq&\frac{1}{k_1}\log\left(1+\frac{\sum_{i=1}^{m-1} e^{-k_1a_{m-1}}}{e^{-k_1a_m}}\right) \nonumber \\
    =&\frac{1}{k_1}\log\left(1+(m-1)e^{-k_1(a_{m-1}-a_m)}\right) \label{Eq:Lm:SmoothMinMaxOperatorErrorBoundsPantStep0}\\ \label{Eq:Lm:SmoothMinMaxOperatorErrorBoundsPantStep1}
    \leq &\ \frac{\log(m)}{k_1}.
\end{align}
Both inequalities used in the above derivation exploit the monotonicity property of the $\log(\cdot)$ function. Using the same property, we can also write
\begin{align}
    (\min\{a\} - \widetilde{\min}\{a\}) 
    =&\ a_m + \frac{1}{k_1}\log\left(\sum_{i=1}^m e^{-k_1a_i}\right) \nonumber \\
    \geq&\ a_m + \frac{1}{k_1}\log\left(e^{-k_1a_m}\right) \nonumber \\ \label{Eq:Lm:SmoothMinMaxOperatorErrorBoundsPantStep2}
    =&\ 0.
\end{align}
The relationships in
\eqref{Eq:Lm:SmoothMinMaxOperatorErrorBoundsPantStep0}, \eqref{Eq:Lm:SmoothMinMaxOperatorErrorBoundsPantStep1} and \eqref{Eq:Lm:SmoothMinMaxOperatorErrorBoundsPantStep2} establish the first result in \eqref{Eq:Lm:SmoothMinMaxOperatorErrorBoundsPant}. Following the same steps, the second result in \eqref{Eq:Lm:SmoothMinMaxOperatorErrorBoundsPant} can also be established. This completes the proof.
\end{proof}

From the above lemma, it is clear that the approximation errors associated with the smooth-min and smooth-max operators in \eqref{Eq:SmoothMinMaxOperatorsPant} are bounded and also vanishes as $k_1,k_2\rightarrow\infty$. Since these operators are used when determining the SRM1 $\tilde{\rho}^{\varphi}(s(u))$ (via Def. \ref{Def:SmoothSTLRobustSemantics}), we can use \eqref{Eq:Lm:SmoothMinMaxOperatorErrorBoundsPant} to bound the approximation error associated with the SRM1 as    
\begin{equation}\label{Eq:Lm:SRM1ApproxError}
    (\rho^{\varphi}(s(u)) - \tilde{\rho}^{\varphi}(s(u)) \in [L_k^\varphi,U_k^\varphi],
\end{equation}
where $L_k^\varphi$ and $U_k^\varphi$ are two real numbers dependent on the STL specification $\varphi$ and the parameters $k = [k_1,k_2]$. In other words, the actual (non-smooth) robustness measure $\rho^{\varphi}(s(u))$ with respect to the SRM1 $\tilde{\rho}^{\varphi}(s(u))$ is placed such that  
\begin{equation}
    \tilde{\rho}^{\varphi}(s(u)) + L_k^\varphi \leq \rho^{\varphi}(y,x,u) \leq \tilde{\rho}^{\varphi}(s(u)) + U_k^\varphi.
\end{equation}
The exact details on how $L_k^\varphi$ and $U_k^\varphi$ can be computed are provided in Section \ref{Sec:ApproximationErrors}. Regardless, at this point, it should be clear from \eqref{Eq:Lm:SmoothMinMaxOperatorErrorBoundsPant} that neither $L_k^\varphi$ nor $U_k^\varphi$ in \eqref{Eq:Lm:SRM1ApproxError} is necessarily zero. This implies that, in this scenario (SRM1), achieving $\tilde{\rho}^{\varphi}(s(u)) > 0$ for some composite signal $s(u)$  does not guarantee the satisfaction of the specification $\varphi$, i.e.,  
\begin{equation*}
    \tilde{\rho}^{\varphi}(s(u)) > 0   \centernot \implies  \rho^{\varphi}(s(u)) > 0.
\end{equation*}
Similarly, in this scenario (SRM1), achieving $\tilde{\rho}^{\varphi}(s(u)) < 0$ for some composite signal $s(u)$  does not guarantee the violation of the specification $\varphi$ either, i.e.,
\begin{equation*}
    \tilde{\rho}^{\varphi}(s(u)) < 0   \centernot \implies  \rho^{\varphi}(s(u)) < 0.
\end{equation*}
The above two characteristics indicate that the SRM1 is neither \emph{sound}, \emph{revers-soundness} nor \emph{complete}. The formal definitions of these concepts: \emph{soundness}, \emph{reverse-soundness} and \emph{completeness} considered in this paper are given in the following definition. 
\begin{definition}\label{Def:SRMProperties}
A smooth robustness measure $\tilde{\rho}^{\varphi}(\cdot)$ is called: \begin{enumerate}
    \item ``\emph{sound}'' if $\tilde{\rho}^{\varphi}(s(u)) > 0   \implies  \rho^{\varphi}(s(u)) > 0$, for any signal $s(u)$ and specification $\varphi$,
    \item ``\emph{reverse-sound}'' if $\tilde{\rho}^{\varphi}(s(u)) < 0   \implies  \rho^{\varphi}(s(u)) < 0$, for any signal $s(u)$ and specification $\varphi$,
    \item ``\emph{complete}'' if it is both sound and reverse-sound. 
\end{enumerate}
\end{definition}

Since the approximation errors associated with the smooth-min and smooth-max operators in \eqref{Eq:SmoothMinMaxOperatorsPant} vanish as the parameters $k_1, k_2 \rightarrow \infty$ (see Lm. \ref{Lm:SmoothMinMaxOperatorErrorBoundsPant}), the same can be expected from the approximation error associated with the SRM1, i.e., as $k_1, k_2 \rightarrow \infty$, the SRM1 $\tilde{\rho}^{\varphi}(s(u)) \rightarrow \rho^{\varphi}(s(u))$. Therefore, as $k_1, k_2 \rightarrow \infty$, SRM1 becomes sound, reverse-sound and complete. Hence, the SRM1 is called \emph{asymptotically sound}, \emph{asymptotically reverse-sound} and \emph{asymptotically complete}.

\subsection{SRM2: Recovering the Soundness Property \cite{Gilpin2021}}

Even though the aforementioned asymptotic properties of SRM1 are theoretically reassuring, when it comes to implementations, allowing the parameters $k_1, k_2 \rightarrow \infty$ is not practical (see \eqref{Eq:SmoothMinMaxOperatorsPant}). This limitation is partially addressed in \cite{Gilpin2021} by proposing a new SRM (labeled as SRM2) that is sound - irrespective of the parameters $k_1, k_2 \in \R_{>0}$. Moreover, this SRM2 is asymptotically reversed-sound and asymptotically complete. In all, SRM2 \cite{Gilpin2021} recovers the lacking soundness property in SRM1 \cite{Pant2017}.  

In particular, SRM2 is defined using $k_1$-Quasi-$\min$ and $k_2$-Soft-$\max$ operators (see Tab. \ref{Tab:SmoothOperators}):
\begin{equation}\label{Eq:SmoothMinMaxOperatorsGilpin}
    \begin{aligned}
        \widetilde{\min}\{a\} = &-\frac{1}{k_1}\log\left(\sum_{i=1}^m e^{-k_1a_i}\right) \mbox{ and}\\
        \widetilde{\max}\{a\} = &\frac{\sum_{i=1}^m a_i e^{k_2a_i}}{\sum_{i=1}^m e^{k_2a_i}},
    \end{aligned}
\end{equation}
as smooth-min and smooth-max operators, respectively. Similar to before, the above operators are smooth everywhere, and their respective gradients are provided in Tab \ref{Tab:SmoothOperators}. The following lemma establishes the bounds on the approximation errors associated with the operators in \eqref{Eq:SmoothMinMaxOperatorsGilpin}.

\begin{lemma}\label{Lm:SmoothMinMaxOperatorErrorBoundsGilpin}
The smooth-min and smooth-max operators in \eqref{Eq:SmoothMinMaxOperatorsGilpin} satisfy the following respective approximation error bounds:
\begin{equation}\label{Eq:Lm:SmoothMinMaxOperatorErrorBoundsGilpin}
    \begin{aligned}
        (\min\{a\} - \widetilde{\min}\{a\})
        \in&[0,\,\frac{1}{k_1}\log(1+(m-1)e^{-k_1(a_{m-1}-a_m)})]\\
        \subseteq&[0,\,\frac{\log(m)}{k_1}] \mbox{ and }\\
        (\max\{a\} - \widetilde{\max}\{a\})
        \in& [0,\,(a_1-a_m)(1-\frac{1}{\sum_{i=1}^m e^{-k_2(a_1-a_i)}})]\\
        \subseteq&[0,\, \frac{a_1-a_m}{1+\frac{e^{k_2(a_1-a_2)}}{m-1}}].
    \end{aligned}
\end{equation}
\end{lemma}
\begin{proof}
The first result in \eqref{Eq:Lm:SmoothMinMaxOperatorErrorBoundsGilpin} has already been proven in Lm. \ref{Lm:SmoothMinMaxOperatorErrorBoundsPant}. Hence, here we only need to prove the second result in \eqref{Eq:Lm:SmoothMinMaxOperatorErrorBoundsGilpin}. The approximation error associated with the smooth-max operator in \eqref{Eq:SmoothMinMaxOperatorsGilpin} is  
\begin{align}
    (\max\{a\} -& \widetilde{\max}\{a\}) 
    = a_1  - \frac{\sum_{i=1}^m a_i e^{k_2a_i}}{\sum_{i=1}^m e^{k_2a_i}} \nonumber \\
    =& \frac{a_1 \sum_{i=2}^m e^{k_2a_i} - \sum_{i=2}^m a_i e^{k_2a_i}}{\sum_{i=1}^m e^{k_2a_i}} \nonumber \\
    \leq& \frac{a_1 \sum_{i=2}^m e^{k_2a_i} - a_m\sum_{i=2}^m e^{k_2a_i}}{\sum_{i=1}^m e^{k_2a_i}} \nonumber \\
    =& (a_1-a_m)\frac{\sum_{i=2}^m e^{k_2a_i}}{\sum_{i=1}^m e^{k_2a_i}} \label{Eq:Lm:SmoothMinMaxOperatorErrorBoundsGilpinStep0} \\
    =& (a_1-a_m)\left(1-\frac{1}{\sum_{i=1}^m e^{-k_2(a_1-a_i)}}\right). \label{Eq:Lm:SmoothMinMaxOperatorErrorBoundsGilpinStep1}
\end{align}
An alternative lower bound can be obtained by continuing from \eqref{Eq:Lm:SmoothMinMaxOperatorErrorBoundsGilpinStep0} as
\begin{align}
    (\max\{a\} - \widetilde{\max}\{a\})  
    \leq& \frac{(a_1-a_m)}{\frac{\sum_{i=1}^m e^{k_2a_i}}{\sum_{i=2}^m e^{k_2a_i}}} 
    = \frac{(a_1-a_m)}{1+\frac{e^{k_2a_1}}{\sum_{i=2}^m e^{k_2a_i}}} \nonumber\\
    \leq& \frac{(a_1-a_m)}{1+\frac{e^{k_2(a_1-a_2)}}{m-1}}. \label{Eq:Lm:SmoothMinMaxOperatorErrorBoundsGilpinStep2}
\end{align}
All the inequalities used in the above derivations exploit the monotonicity property the exponential function. Using the same property, we can also write
\begin{align}
    (\max\{a\} - \widetilde{\max}\{a\})
     =&\ a_1  - \frac{\sum_{i=1}^m a_i e^{k_2a_i}}{\sum_{i=1}^m e^{k_2a_i}} \nonumber \\
    \geq&\ a_1  - \frac{a_1\sum_{i=1}^m e^{k_2a_i}}{\sum_{i=1}^m e^{k_2a_i}}\nonumber\\
    =&\ 0.
    \label{Eq:Lm:SmoothMinMaxOperatorErrorBoundsGilpinStep3}
\end{align}
The relationships in
\eqref{Eq:Lm:SmoothMinMaxOperatorErrorBoundsGilpinStep1}, \eqref{Eq:Lm:SmoothMinMaxOperatorErrorBoundsGilpinStep2} and \eqref{Eq:Lm:SmoothMinMaxOperatorErrorBoundsGilpinStep3} prove the second result in \eqref{Eq:Lm:SmoothMinMaxOperatorErrorBoundsGilpin}, and thus, the proof is complete.
\end{proof}

Recall that the operators in \eqref{Eq:SmoothMinMaxOperatorsGilpin} define the SRM2 $\tilde{\rho}^{\varphi}(s(u))$ via Def. \ref{Def:SmoothSTLRobustSemantics}. Now, the following lemma can be established regarding the approximation error associated with the SRM2.

\begin{lemma}\label{Lm:SRM2ApproxError}
The approximation error associated with the SRM2 $\tilde{\rho}^{\varphi}(s(u))$ is bounded such that
\begin{equation}\label{Eq:Lm:SRM2ApproxError}
    (\rho^{\varphi}(s(u)) - \tilde{\rho}^{\varphi}(s(u))) \in [0,U_k^\varphi],
\end{equation}
where $U_k^\varphi$ is some positive real number dependent on the specification $\varphi$ and the parameters $k = [k_1,k_2]$, and $U_k^\varphi \rightarrow 0$ as $k_1,k_2 \rightarrow \infty$.  
\end{lemma}
\begin{proof}
Lemma \ref{Lm:SmoothMinMaxOperatorErrorBoundsGilpin} implies that for any set of real numbers $a$, the smooth operators in \eqref{Eq:SmoothMinMaxOperatorsGilpin} satisfy $\widetilde{\min}\{a\} \leq \min\{a\}$ and $\widetilde{\max}\{a\} \leq \max\{a\}$. Note also that SRM2 $\tilde{\rho}^{\varphi}(s(u))$ is computed recursively using the semantics in Def. \ref{Def:SmoothSTLRobustSemantics} while $\rho^{\varphi}(s(u))$ is computed recursively using the semantics in Def. \ref{Def:STLRobustSemantics}. Then, from comparing the respective semantics in Defs. \ref{Def:STLRobustSemantics} and \ref{Def:SmoothSTLRobustSemantics}, it is easy to see that $\tilde{\rho}^{\varphi}(s(u)) \leq \rho^{\varphi}(s(u))$ as   
\begin{itemize}
    \item $\tilde{\rho}^\pi((s,t)) = \rho^\pi((s,t))$
    \item $\tilde{\rho}^{\neg\pi}((s,t)) = \rho^{\neg\pi}((s,t))$
    \item $\tilde{\rho}^{\varphi_1 \land \varphi_2}((s,t)) \leq \rho^{\varphi_1 \land \varphi_2}((s,t))$
    \item $\tilde{\rho}^{\varphi_1 \lor \varphi_2}((s,t)) \leq \rho^{\varphi_1 \lor \varphi_2}((s,t))$ 
    \item $\tilde{\rho}^{\mathbf{F}_{[t_1,t_2]}\varphi}((s,t))\leq \rho^{\mathbf{F}_{[t_1,t_2]}\varphi}((s,t))$
    \item $\tilde{\rho}^{\mathbf{G}_{[t_1,t_2]}\varphi}((s,t))\leq \rho^{\mathbf{G}_{[t_1,t_2]}\varphi}((s,t))$
    \item $\tilde{\rho}^{\varphi_1\mathbf{U}_{[t_1,t_2]}\varphi_2}((s,t))\leq \rho^{\varphi_1\mathbf{U}_{[t_1,t_2]}\varphi_2}((s,t))$.
\end{itemize}
Moreover, according to Lm. \ref{Lm:SmoothMinMaxOperatorErrorBoundsGilpin}, the approximation errors associated with the operators in \eqref{Eq:SmoothMinMaxOperatorsGilpin} are bounded and also vanishes as $k_1,k_2\rightarrow \infty$. Therefore, the approximation error associated with the SRM2 should follow \eqref{Eq:Lm:SRM2ApproxError}. 
\end{proof}

The above lemma implies that the actual (non-smooth) robustness measure $\rho^{\varphi}(s(u))$ with respect to the SRM2 $\tilde{\rho}^{\varphi}(s(u))$ is placed such that 
\begin{equation*}
    \tilde{\rho}^{\varphi}(s(u)) \leq \rho^{\varphi}(s(u)) \leq \tilde{\rho}^{\varphi}(s(u)) + U_k^\varphi.
\end{equation*}
Therefore, in this scenario (SRM2), achieving $\tilde{\rho}^{\varphi}(s(u)) > 0$  guarantees the satisfaction of the specification, i.e.,
\begin{equation*}
    \tilde{\rho}^{\varphi}(s(u)) > 0 \implies \rho^{\varphi}(s(u)) > 0 \implies s(u) \vDash \varphi.
\end{equation*}
This indeed is the soundness property promised earlier. Hence, it is clear that the SRM2 proposed in \cite{Gilpin2021} is sound and also asymptotically reverse-sound and asymptotically complete.

\subsection{SRM3: Recovering the reverse-soundness}

Even though having a sound SRM (e.g., SRM2) is important in \emph{verification} type problems, in \emph{falsification} type problems (where the aim is to find a control that violates a given specification), it is critical to have a reverse-sound SRM. 
Further, since a reverse-sound SRM will, by definition, over-approximate the actual robustness measure (as opposed to a sound SRM that under-approximates the actual robustness measure), one can expect to use a sound SRM together with a reverse-sound SRM (e.g., use one as a ``boosting'' objective function for the other \cite{Welikala2019J1}) to achieve better and more decisive results. Motivated by these prospects, in the same spirit of the SRM2, here we propose a new SRM (labeled as SRM3) that recovers the lacking reverse-soundness property in SRM1 \cite{Pant2017}, while also preserving the asymptotic soundness and asymptotic completeness properties in SRM1 \cite{Pant2017}.

The proposed SRM3 is defined using $k_1$-Soft-$\min$ and $k_2$-Quasi-$\max$ operators (see Tab. \ref{Tab:SmoothOperators}):   
\begin{equation}\label{Eq:SmoothMinMaxOperatorsGilpinReversed}
    \begin{aligned}
        \widetilde{\min}\{a\} = &\frac{\sum_{i=1}^m a_i e^{-k_1a_i}}{\sum_{i=1}^m e^{-k_1a_i}} \mbox{ and}\\
        \widetilde{\max}\{a\} = &\frac{1}{k_2}\log\left(\sum_{i=1}^m e^{k_2a_i}\right),
    \end{aligned}
\end{equation}
as smooth-min and smooth-max operators, respectively. Similar to before, the above operators are smooth everywhere, and their respective gradients can be found in Tab. \ref{Tab:SmoothOperators}. The following lemma establishes the bounds on the approximation error associates with the operators in \eqref{Eq:SmoothMinMaxOperatorsGilpinReversed}.

\begin{lemma}\label{Lm:SmoothMinMaxOperatorErrorBoundsGilpinReversed}
The smooth-min and smooth-max operators in \eqref{Eq:SmoothMinMaxOperatorsGilpinReversed} satisfy the following respective approximation error bounds:
\begin{equation}\label{Eq:Lm:SmoothMinMaxOperatorErrorBoundsGilpinReversed}
    \begin{aligned}
        (\min\{a\} - \widetilde{\min}\{a\}) \in&[-(a_1-a_m)(1-\frac{1}{\sum_{i=1}^m e^{-k_1(a_i-a_m)}}),\,0]\\ \subseteq&[-\frac{a_1-a_m}{1+\frac{e^{k_1(a_{m-1}-a_m)}}{m-1}},\,0] \mbox{ and }\\
        (\max\{a\} - \widetilde{\max}\{a\}) \in&[-\frac{1}{k_2}\log(1+(m-1)e^{-k_2(a_1-a_2)}),\,0]\\ \subseteq& [-\frac{\log(m)}{k_2},\,0]
    \end{aligned}
\end{equation}
\end{lemma}
\begin{proof}
The first result in \eqref{Eq:Lm:SmoothMinMaxOperatorErrorBoundsGilpinReversed} can be proved by following the same steps used in the proof of the second statement in \eqref{Eq:Lm:SmoothMinMaxOperatorErrorBoundsGilpin} in  Lm. \ref{Lm:SmoothMinMaxOperatorErrorBoundsGilpin}. The second result in \eqref{Eq:Lm:SmoothMinMaxOperatorErrorBoundsGilpinReversed} has already been proven in Lm. \ref{Lm:SmoothMinMaxOperatorErrorBoundsPant}. 
\end{proof}

In parallel to Lm. \ref{Lm:SRM2ApproxError}, the following lemma can be established regarding the approximation error associated with the SRM3 $\tilde{\rho}^{\varphi}(s(u))$ defined by the operators in \eqref{Eq:SmoothMinMaxOperatorsGilpinReversed}. 

\begin{lemma}\label{Lm:SRM3ApproxError}
The approximation error associated with the SRM3 $\tilde{\rho}^{\varphi}(s(u))$ is bounded such that
\begin{equation}\label{Eq:Lm:SRM3ApproxError}
    (\rho^{\varphi}(s(u)) - \tilde{\rho}^{\varphi}(s(u))) \in [L_k^\varphi,0],
\end{equation}
where $L_k^\varphi$ is some negative real number dependent on the specification $\varphi$ and the parameters $k = [k_1,k_2]$, and $L_k^\varphi \rightarrow 0$ as $k_1,k_2 \rightarrow \infty$.  
\end{lemma}
\begin{proof}
The proof follows the same steps as that of Lm. \ref{Lm:SRM2ApproxError}, and is, therefore, omitted.
\end{proof}

The above lemma implies that the actual (non-smooth) robustness measure $\rho^{\varphi}(s(u))$ lies in the range:
\begin{equation*}
    \tilde{\rho}^{\varphi}(s(u)) + L_k^\varphi \leq \rho^{\varphi}(s(u)) \leq \tilde{\rho}^{\varphi}(s(u)),
\end{equation*}
Hence, in this scenario (SRM3), achieving $\tilde{\rho}^{\varphi}(s(u)) < 0$ implies a violation of the specification (i.e., $s(u) \nvDash \varphi$). i.e.,
\begin{equation*}
    \tilde{\rho}^{\varphi}(s(u)) < 0 \implies \rho^{\varphi}(s(u)) < 0 \implies s(u) \nvDash \varphi.
\end{equation*}
This indeed is the reverse-soundness property promised earlier. Therefore, it is clear that the proposed SRM3 is reverse sound and also asymptotically sound and asymptotically complete.

\subsection{SRM4: Lifting the uniform robustness}

As a complementary approach for the SRM1 \cite{Pant2017}, we now propose a new SRM (labeled SRM4). In particular, this SRM4 is defined using the $k_1$-Soft-$\min$ and $k_2$-Soft-$\max$ operators (see Tab \ref{Tab:SmoothOperators}):
\begin{equation}\label{Eq:SmoothMinMaxOperatorsPantReversed}
    \begin{aligned}
        \widetilde{\min}\{a\} = &\frac{\sum_{i=1}^m a_i e^{-k_1a_i}}{\sum_{i=1}^m e^{-k_1a_i}} \mbox{ and}\\
        \widetilde{\max}\{a\} = &\frac{\sum_{i=1}^m a_i e^{k_2a_i}}{\sum_{i=1}^m e^{k_2a_i}},
    \end{aligned}
\end{equation}
as smooth-min and smooth-max operators, respectively. 

Compared to the smooth-min and smooth-max operators used in SRM1 (given in \eqref{Eq:SmoothMinMaxOperatorsPant}), the operators in \eqref{Eq:SmoothMinMaxOperatorsPantReversed} have a special property: when $k_1,k_2 \rightarrow 0$, both operators converge to the arithmetic mean operator. According to the findings in \cite{Mehdipour2019}, optimizing a robustness measure computed based on mean-like (e.g., arithmetic mean) operators rather than min/max-like (e.g., smooth-min/max) operators can lead to more uniformly robust solutions. This is because, while the latter (min/max-based) approach optimizes the minimum margin of satisfaction of all the specification components (e.g., predicates), the former (mean-based) approach optimizes the mean margin of satisfaction of all the specification components. Hence, one can expect to use the SRM1  (with moderate $k_1,k_2$ values) together with the SRM4 (with low $k_1,k_2$ values) to reach more uniformly robust solutions (e.g., use SRM4 as a ``boosting'' objective function for the SRM1 \cite{Welikala2019J1}). Even when used alone, intuitively, the SRM4 (with moderate $k_1,k_2$ values) can be expected to provide more uniformly robust solutions.

Similar to previous SRMs, the above operators are smooth everywhere, and their respective gradients are provided in Tab \ref{Tab:SmoothOperators}. The following lemma establishes the bounds on the approximation errors associated with the operators in \eqref{Eq:SmoothMinMaxOperatorsPantReversed}.

\begin{lemma}\label{Lm:SmoothMinMaxOperatorErrorBoundsPantReversed}
The smooth-min and smooth-max operators in \eqref{Eq:SmoothMinMaxOperatorsPantReversed} satisfy the following respective approximation error bounds:
\begin{equation}\label{Eq:Lm:SmoothMinMaxOperatorErrorBoundsPantReversed}
    \begin{aligned}
        (\min\{a\} - \widetilde{\min}\{a\})
        \in&[-(a_1-a_m)(1-\frac{1}{\sum_{i=1}^m e^{-k_1(a_i-a_m)}}),\,0]\\ \subseteq&[-\frac{a_1-a_m}{1+\frac{e^{k_1(a_{m-1}-a_m)}}{m-1}},\,0] \mbox{ and }\\
        (\max\{a\} - \widetilde{\max}\{a\})
        \in& [0,\,(a_1-a_m)(1-\frac{1}{\sum_{i=1}^m e^{-k_2(a_1-a_i)}})]\\
        \subseteq&[0,\, \frac{a_1-a_m}{1+\frac{e^{k_2(a_1-a_2)}}{m-1}}].
    \end{aligned}
\end{equation}
\end{lemma}
\begin{proof}
The both results in \eqref{Eq:Lm:SmoothMinMaxOperatorErrorBoundsPantReversed} has already been proven in Lms. \ref{Lm:SmoothMinMaxOperatorErrorBoundsGilpin} and \ref{Lm:SmoothMinMaxOperatorErrorBoundsGilpinReversed}.
\end{proof}

From the above lemma, it is clear that the approximation errors associated with the smooth-min and smooth-max operators in \eqref{Eq:SmoothMinMaxOperatorsPantReversed} are bounded and also vanishes as $k_1,k_2\rightarrow\infty$. Therefore, similar to the case with the SRM1, the approximation error associated with the SRM4 can be bounded as
\begin{equation}\label{Eq:Lm:SRM4ApproxError}
    (\rho^{\varphi}(s(u)) - \tilde{\rho}^{\varphi}(s(u)) \in [L_k^\varphi,U_k^\varphi], 
\end{equation}
i.e., the actual robustness measure $\rho^{\varphi}(s(u))$ is such that
\begin{equation}
    \tilde{\rho}^{\varphi}(s(u)) + L_k^\varphi \leq \rho^{\varphi}(y,x,u) \leq \tilde{\rho}^{\varphi}(s(u)) + U_k^\varphi,
\end{equation}
where $L_k^\varphi$ and $U_k^\varphi$ are two real numbers dependent on the STL specification $\varphi$ and the parameters $k = [k_1,k_2]$. The exact details on how $L_k^\varphi$ and $U_k^\varphi$ can be computed are provided in Section \ref{Sec:ApproximationErrors}. Regardless, it should be clear that $k_1, k_2 \rightarrow \infty \implies L_k^\varphi, U_k^\varphi \rightarrow 0$. Therefore, similar to the SRM1, the SRM4 is asymptotically sound, asymptotically reverse-sound and asymptotically complete.

\section{Approximation Errors of Smooth Robustness Measures} 
\label{Sec:ApproximationErrors}

In the previous section, we discussed four candidate smooth robustness measures $\tilde{\rho}^{\varphi}(s)$ (SRM1-4) that can be used to approximate the actual non-smooth robustness measure $\rho^{\varphi}(s)$ when solving the synthesis problem \eqref{Eq:SynthesisProblem3} via a gradient-based optimization method. For notational convenience, let us define the corresponding approximation error as 
\begin{equation}
    \tilde{e}^\varphi(s) \triangleq (\rho^\varphi(s) - \tilde{\rho}^\varphi(s)). 
\end{equation} 
As shown in \eqref{Eq:Lm:SRM1ApproxError}, \eqref{Eq:Lm:SRM2ApproxError}, \eqref{Eq:Lm:SRM3ApproxError} and  \eqref{Eq:Lm:SRM4ApproxError} (obtained via Lms. \ref{Lm:SmoothMinMaxOperatorErrorBoundsPant}, \ref{Lm:SmoothMinMaxOperatorErrorBoundsGilpin}, \ref{Lm:SmoothMinMaxOperatorErrorBoundsGilpinReversed} and \ref{Lm:SmoothMinMaxOperatorErrorBoundsPantReversed}, respectively), this approximation error $\tilde{e}^\varphi(s)$ will always be bounded inside some finite interval (henceforth called an \emph{error bound}) of the form:
\begin{equation}\label{Eq:ApproximationErrorBandDefinition}
    \tilde{e}^\varphi(s) \in [L_\theta^\varphi, U_\theta^\varphi] \subset \R,
\end{equation}
irrespective of the signal $s=s(u)$. Here, $L_\theta^\varphi$ and $U_\theta^\varphi$ are two real numbers where $\theta$ can be thought of as a collection of all the used smooth-min and smooth-max operator parameters (e.g., $\theta = [k_1,k_2]$). Intuitively, this error bound $[L_\theta^\varphi, U_\theta^\varphi]$ will depend on: 1) the specification $\varphi$, 
2) the SRM (i.e., the used smooth-min and smooth-max operator configuration: \eqref{Eq:SmoothMinMaxOperatorsPant}, \eqref{Eq:SmoothMinMaxOperatorsGilpin}, \eqref{Eq:SmoothMinMaxOperatorsGilpinReversed} or \eqref{Eq:SmoothMinMaxOperatorsPantReversed}), and 
3) the smooth-min and smooth-max operator parameters $\theta$.

In this section, we develop a set of semantics (parallel to the ones in Def. \ref{Def:SmoothSTLRobustSemantics}) to explicitly derive the error bound $[L_\theta^\varphi, U_\theta^\varphi]$ in \eqref{Eq:ApproximationErrorBandDefinition} for a given specification and an SRM - in terms of the operator parameters $\theta$. Having this knowledge is critical as it enables one to sensibly select: 1) the SRM to optimize (in \eqref{Eq:SynthesisProblem3}), and 2) the operator parameters $\theta$. For example, one can select an SRM and its operator parameters $\theta$ that minimizes the width of the error bound, i.e.,  $U_\theta^\varphi-L_\theta^\varphi$.

\subsection{Some Preliminaries}
In the following analysis, we use the notation 
\begin{equation}\label{Eq:ApproxErrorBandsOfSmoothOperators}
    \begin{aligned}
        (\min\{a\} - \widetilde{\min}\{a\} ) \in& [L_{k_1,m}^{\min},\ U_{k_1,m}^{\min}] \mbox{ and }\\
        (\max\{a\} - \widetilde{\max}\{a\} ) \in& [L_{k_2,m}^{\max},\ U_{k_2,m}^{\max}],\\
    \end{aligned}
\end{equation}
to commonly represent the approximation error bounds associated with different possible respective smooth-min and smooth-max operators proven in Lms. \ref{Lm:SmoothMinMaxOperatorErrorBoundsPant}, \ref{Lm:SmoothMinMaxOperatorErrorBoundsGilpin}, \ref{Lm:SmoothMinMaxOperatorErrorBoundsGilpinReversed}, \ref{Lm:SmoothMinMaxOperatorErrorBoundsPantReversed} (also outlined in Tab. \ref{Tab:SmoothOperators}). Note also that, in \eqref{Eq:ApproxErrorBandsOfSmoothOperators}, $m$ is the cardinality of the set $\bar{a}$ and $k_1,k_2\in\R_{>0}$ are operator parameters that a user needs to select. For example, if the SRM1 is to be used,
\begin{align*}
[L_{k_1,m}^{\min},\ U_{k_1,m}^{\min}] \equiv& [0,\ \frac{\log(m)}{k_1}], \mbox{ and }\\
[L_{k_2,m}^{\max},\ U_{k_2,m}^{\max}] \equiv& [-\frac{\log(m)}{k_2},\ 0].\\
\end{align*}

We also require the following minor theoretical result.
\begin{lemma}\label{Lm:Fundamentalinequality}
If $a_i,\tilde{a}_i,L_i,U_i \in\R$ are such that $(a_i - \tilde{a}_i)\in[L_i, U_i]$ for all $i\in [1,m]\ (\triangleq \{1,2,\ldots,m\})$, then 
\begin{equation}\label{Eq:Lm:Fundamentalinequality}
    \begin{aligned}
    \left(\underset{i\in [1,m]}{\min} \{a_i\} - \underset{i\in [1,m]}{\min} \{\tilde{a}_i\}\right) 
\in \left[\underset{i \in [1,m]}{\min} \{L_i\},\ \ \underset{i \in [1,m]}{\max} \{U_i\}\right],\\
\left(\underset{i\in [1,m]}{\max} \{a_i\} - \underset{i\in [1,m]}{\max} \{\tilde{a}_i\}\right) 
\in \left[\underset{i \in [1,m]}{\min} \{L_i\},\ \ \underset{i \in [1,m]}{\max} \{U_i\}\right]. 
    \end{aligned}
\end{equation}
\end{lemma}
\begin{proof}
To prove the first result, take $a_j = \min_{i\in [1,m]} \{a_i\}$ and $\tilde{a}_k = \min_{i\in [1,m]} \{\tilde{a}_i\}$ where $j,k\in [1,m]$. According to this definition, $a_j \leq a_k$ and $\tilde{a}_k \leq \tilde{a}_j$. Also, according to the given information, $(a_j-\tilde{a}_j) \geq L_j$ and $(a_k - \tilde{a}_k) \leq U_k$. We now can use these relationships to impose bounds on $(a_j-\tilde{a}_k)$ as: 
\begin{align*}
     (a_j-\tilde{a}_k) 
     &\geq (a_j-\tilde{a}_j) \geq L_j \geq \min_{i\in[1,m]}\{L_i\}, \mbox{ and }\\
     (a_j-\tilde{a}_k) 
     &\leq  (a_k - \tilde{a}_k) \leq U_k \leq \max_{i\in[1,m]}\{U_i\}.
\end{align*}
This proves the first result in \eqref{Eq:Lm:Fundamentalinequality}. Following the same steps, the second result in \eqref{Eq:Lm:Fundamentalinequality} can be established.
\end{proof}


Finally, we improve the generality of the proposing set of semantics (that computes the error bound $[L_\theta^\varphi, U_\theta^\varphi]$ in \eqref{Eq:ApproximationErrorBandDefinition}) by introducing a small modification to the smooth STL robust semantics given in Def. \ref{Def:SmoothSTLRobustSemantics}. In particular, we replace the first smooth STL robust semantic: $\tilde{\rho}^\pi(s) = \mu^\pi(s_t)$ by a noise affected version of it:
\begin{equation}\label{Eq:ModifiedSmoothSTLRobustSemantic}
    \tilde{\rho}^\pi((s,t)) = \mu^\pi(s_t) + w^\pi(s_t),
\end{equation}
where the noise term $w^\pi(s_t)$ is such that $w^\pi(s_t) \in [L^\pi(s_t),U^\pi(s_t)]$ with known functions $L^\pi:\R^p \rightarrow \R$ and $U^\pi:\R^p \rightarrow \R$. For example, $L^\pi(s_t) = -U^\pi(s_t) = -\epsilon$ where $\epsilon\in\R_{>0}$ is a known constant, is a possibility. Intuitively, this modification causes the error bound $[L_\theta^\varphi, U_\theta^\varphi]$ in \eqref{Eq:ApproximationErrorBandDefinition} to become dependent on the signal $s$, i.e., now,
\begin{equation}\label{Eq:ApproximationErrorBandDefinition2}
    \tilde{e}^\varphi(s) \triangleq (\rho^\varphi(s) - \tilde{\rho}^\varphi(s)) \in [L_\theta^\varphi(s), U_\theta^\varphi(s)].
\end{equation}
Therefore, this modification complicates the task at hand: computing the error bound $[L_\theta^\varphi(s), U_\theta^\varphi(s)]$. However, as we will see in the sequel, it allows us to see how the uncertainties in the mission space can affect this error bound $[L_\theta^\varphi(s), U_\theta^\varphi(s)]$.

\begin{remark}\label{Rm:NoMissionSpaceUncertainities}
If there are no uncertainties in the mission space, one can simply set $L^\pi(s_t) = U^\pi(s_t) = 0$ in \eqref{Eq:ModifiedSmoothSTLRobustSemantic} for all $s_t$ and for all predicates $\pi$. Moreover, if such uncertainties are independent of the agent trajectory $s$, one can simply set $L^\pi(s_t) = L^\pi$ and $U^\pi(s_t) = U^\pi$ where $L^\pi$ and $U^\pi$ are known constants. In both of these occasions, the resulting error bound $[L_\theta^\varphi(s), U_\theta^\varphi(s)]$ will be independent of the signal $s$. Evaluating such a signal independent error bound is a reasonable thing to do, especially in an off-line stage where the signal $s$ is unknown and the focus is on sensibly selecting an SRM along with its operator parameters $\theta$.
\end{remark}

\subsection{STL error semantics}

Consider the set of semantics (we named the \emph{STL error semantics}) defined in the following definition.  

\begin{definition}\label{Def:STLErrorSemantics}
(STL Error Semantics)
\begin{itemize}
\item  
$\tilde{e}^\pi((s,t)) \in [-U^\pi(s_t),\ -L^\pi(s_t)]$
\item 
$\tilde{e}^{\neg \pi}((s,t)) \in [L^\pi(s_t),\ U^\pi(s_t)]$
\item 
$\tilde{e}^{\varphi_1 \land \varphi_2}((s,t)) \in 
[L_{k_1,2}^{\min}+\min\{L_{\theta_1}^{\varphi_1}((s,t)),L_{\theta_2}^{\varphi_2}((s,t))\},\ $
\item[] 
\hfill $U_{k_1,2}^{\min}+\max\{U_{\theta_1}^{\varphi_1}((s,t)),U_{\theta_2}^{\varphi_2}((s,t))\}]$
\item 
$\tilde{e}^{\varphi_1 \lor \varphi_2}((s,t)) \in [ 
L_{k_2,2}^{\max}+\min\{L_{\theta_1}^{\varphi_1}((s,t)),L_{\theta_2}^{\varphi_2}((s,t))\},\ $
\item[] 
\hfill $U_{k_2,2}^{\max}+\max\{U_{\theta_1}^{\varphi_1}((s,t)),U_{\theta_2}^{\varphi_2}((s,t))\}]$
\item
$\tilde{e}^{\mathbf{F}_{[t_1,t_2]}\varphi}((s,t)) \in [L_{k_2,t_2-t_1+1}^{\max} + \min_{\tau\in[t+t_1,t+t_2]}\{$ 
\item[]
\hfill $L_\theta^{\varphi}((s,\tau))\},\ U_{k_2,t_2-t_1+1}^{\max}+\max_{\tau\in[t+t_1,t+t_2]}\{U_\theta^{\varphi}((s,\tau))\}]$
\item 
$\tilde{e}^{\mathbf{G}_{[t_1,t_2]}\varphi}((s,t)) 
\in [L_{k_1,t_2-t_1+1}^{\min}+\min_{\tau\in[t+t_1,t+t_2]}\{$
\item[] 
\hfill $L_\theta^{\varphi}((s,\tau))\},\ U_{k_1,t_2-t_1+1}^{\min}+\max_{\tau\in[t+t_1,t+t_2]}\{U_\theta^{\varphi}((s,\tau))\}]$
\item 
$\tilde{e}^{\varphi_1\mathbf{U}_{[t_1,t_2]}\varphi_2}((s,t)) 
\in [L_{k_2,t_2-t_1+1}^{\max} + \min_{\tau\in[t+t_1,t+t_2]}\{$
\item[] 
$L_{\theta_4}^{\varphi_4}((s,\tau))\},U_{k_2,t_2-t_1+1}^{\max} + \max_{\tau\in[t+t_1,t+t_2]}\{U_{\theta_4}^{\varphi_4}((s,\tau))\}]$
\item[] where
\item[]
$[L_{\theta_4}^{\varphi_4}((s,\tau)), U_{\theta_4}^{\varphi_4}((s,\tau))] \equiv\ 
[L_{k_1,2}^{\min} + \min\{L_{\theta_1}^{\varphi_1}((s,\tau)),$
\item[]
\hfill $L_{\theta_3}^{\varphi_3}((s,\tau))\}, U_{k_1,2}^{\min} + \max\{U_{\theta_1}^{\varphi_1}((s,\tau)),U_{\theta_3}^{\varphi_3}((s,\tau))\}],$
\item[]
$[L_{\theta_3}^{\varphi_3}((s,\tau)),U_{\theta_3}^{\varphi_3}((s,\tau))]\equiv\ 
[L_{k_1,\tau-t-t_1+1} + \min_{\delta\in[t+t_1,\tau]}$
\item[]
\hfill $\{L_{\theta_2}^{\varphi_2}((s,\delta))\},
U_{k_1,\tau-t-t_1+1} + \max_{\delta\in[t+t_1,\tau]}\{U_{\theta_2}^{\varphi_2}((s,\delta))
\}]$.
\end{itemize}
\end{definition}

\begin{theorem}\label{Th:STLErrorSemantics}
For a known signal suffix $(s,t)$, a specification $\varphi$ and a smooth robustness measure (SRM1-4), the \emph{STL error semantics} in Def. \ref{Def:SmoothSTLRobustSemantics} can be used to determine the approximation error bound 
$$[L_\theta^\varphi((s,t)), U_\theta^\varphi((s,t))] \ni \tilde{e}^\varphi((s,t)) \triangleq (\rho^\varphi((s,t))-\tilde{\rho}^\varphi((s,t))),$$ 
(the same as \eqref{Eq:ApproximationErrorBandDefinition2}) in terms of the used collection of smooth-min and smooth-max operator parameters $\theta$.
\end{theorem}
\begin{proof}
Here, we need to prove each semantic in Def. \ref{Def:STLErrorSemantics} using the respective semantics in Defs. \ref{Def:STLRobustSemantics} and \ref{Def:SmoothSTLRobustSemantics}.  

\paragraph*{Semantic 1} According to Defs. \ref{Def:STLRobustSemantics} and \ref{Def:SmoothSTLRobustSemantics}, $\rho^\pi((s,t)) = \mu^\pi(s_t)$ and $\tilde{\rho}^\pi((s,t)) = \mu^\pi(s_t) + w^\pi(s_t)$ (recall \eqref{Eq:ModifiedSmoothSTLRobustSemantic}), respectively. Therefore, the corresponding approximation error is 
$
\tilde{e}^\pi((s,t)) \triangleq (\rho^\pi((s,t)) - \tilde{\rho}^\pi((s,t))) = -w^\pi(s_t).
$
Since $w^\pi(s_t) \in [L^\pi(s_t),U^\pi(s_t)]$ (from \eqref{Eq:ModifiedSmoothSTLRobustSemantic}), it is easy to see that 
$\tilde{e}^\pi((s,t))\in [-U^\pi(s_t),-L^\pi(s_t)]$ (i.e., the Semantic 1).

\paragraph*{Semantic 2} Can be proved using the same steps as above.

\paragraph*{Semantic 3} \hfill According to the corresponding semantics in Defs. \ref{Def:STLRobustSemantics} and \ref{Def:SmoothSTLRobustSemantics}, 
$\rho^{\varphi_1 \land \varphi_2}((s,t)) = \min \{\rho^{\varphi_1}((s,t)),\rho^{\varphi_2}((s,t))\}$ and 
$\tilde{\rho}^{\varphi_1 \land \varphi_2}((s,t)) = \widetilde{\min} \{\tilde{\rho}^{\varphi_1}((s,t)),\tilde{\rho}^{\varphi_2}((s,t))\}$, respectively. Therefore, the corresponding approximation error is  
$\tilde{e}^{\varphi_1 \land \varphi_2}((s,t)) 
\triangleq (\rho^{\varphi_1 \land \varphi_2}((s,t)) - \tilde{\rho}^{\varphi_1 \land \varphi_2}((s,t))) =  
\min \{\rho^{\varphi_1}((s,t)),\, \rho^{\varphi_2}((s,t))\} - \widetilde{\min} \{\tilde{\rho}^{\varphi_1}((s,t)),\, \tilde{\rho}^{\varphi_2}((s,t))\}
$.

To bound this error term $\tilde{e}^{\varphi_1 \land \varphi_2}((s,t))$, what we have at our disposal are the error bounds:
$
[L_{\theta_i}^{\varphi_i}((s,t)), U_{\theta_i}^{\varphi_i}((s,t))] \ni \tilde{e}^{\varphi_i}((s,t)) \triangleq  (\rho^{\varphi_i}((s,t)) - \tilde{\rho}^{\varphi_i}((s,t)))
$,
for $i=1,2.$ Applying these two error bounds in Lm \ref{Lm:Fundamentalinequality}, we get
\begin{equation}\label{Eq:ApproxErrorBandSemantic3Step1}
\begin{aligned}
    \min \{\rho^{\varphi_1}((s,t)),\, \rho^{\varphi_2}((s,t))\} - \min \{\tilde{\rho}^{\varphi_1}((s,t)),\, \tilde{\rho}^{\varphi_2}((s,t))\} \\
    \in [\min \{L_{\theta_1}^{\varphi_1}((s,t)),L_{\theta_2}^{\varphi_2}((s,t))\},\\ \max\{U_{\theta_1}^{\varphi_1}((s,t)),U_{\theta_2}^{\varphi_2}((s,t))\}].
\end{aligned}
\end{equation}

Further, using \eqref{Eq:ApproxErrorBandsOfSmoothOperators}, we can write 
\begin{equation}\label{Eq:ApproxErrorBandSemantic3Step2}
\begin{aligned}
    \min\{\tilde{\rho}^{\varphi_1}((s,t)), \tilde{\rho}^{\varphi_2}((s,t))\}
    - \widetilde{\min} \{\tilde{\rho}^{\varphi_1}((s,t)),\, \tilde{\rho}^{\varphi_2}((s,t))\} \\
    \in [L_{k_1,2}^{\min},U_{k_1,2}^{\min}].
\end{aligned}
\end{equation}

Finally, by adding \eqref{Eq:ApproxErrorBandSemantic3Step1} and \eqref{Eq:ApproxErrorBandSemantic3Step2} we can obtain a bound for the approximation error term $\tilde{e}^{\varphi_1 \land \varphi_2}((s,t))$ as  
\begin{equation*}
\begin{aligned}
\tilde{e}^{\varphi_1 \land \varphi_2}((s,t)) \in&\ 
[L_{\theta}^{\varphi_1 \land \varphi_2}((s,t)),\ U_{\theta}^{\varphi_1 \land \varphi_2}((s,t))]\\
\equiv&\ \lbrack L_{k_1,2}^{\min}+\min\{L_{\theta_1}^{\varphi_1}((s,t)),L_{\theta_2}^{\varphi_2}((s,t))\},\\ 
&\ \ U_{k_1,2}^{\min}+\max\{U_{\theta_1}^{\varphi_1}((s,t)),U_{\theta_2}^{\varphi_2}((s,t))\}
],
\end{aligned}
\end{equation*}
where $\theta = \theta_1 \cup \theta_2 \cup \{k_1\}$ (which makes it clear how the collection of operator parameters $\theta$ grows with the depth of the specification). This proves the Semantic 3.

\paragraph*{Semantic 4} Can be proved using the same steps as above.

\paragraph*{Semantic 5} Similar to before, by subtracting the corresponding semantics in Defs. \ref{Def:STLRobustSemantics}, \ref{Def:SmoothSTLRobustSemantics}, we get the error term
\begin{equation}\label{Eq:ApproxErrorBandSemantic5Step1}
\begin{aligned}
    \tilde{e}^{\mathbf{F}_{[t_1,t_2]}\varphi}((s,t)) \triangleq
    \rho^{\mathbf{F}_{[t_1,t_2]}\varphi}((s,t)) - \tilde{\rho}^{\mathbf{F}_{[t_1,t_2]}\varphi}((s,t))\\
    = \max_{\tau\in[t+t_1,t+t_2]}\{\rho^\varphi((s,\tau))\} - \underset{\tau\in[t+t_1,t+t_2]}{\widetilde{\max}}\{\tilde{\rho}^\varphi((s,\tau))\}.
\end{aligned}
\end{equation}

To bound this error term $\tilde{e}^{\mathbf{F}_{[t_1,t_2]}\varphi}((s,t))$, what we have at our disposal are the bounds: 
$
[L_\theta^{\varphi}((s,\tau)),  U_\theta^{\varphi}((s,\tau))] \ni \tilde{e}^{\varphi}((s,\tau)) \triangleq (\rho^{\varphi}((s,\tau)) - \tilde{\rho}^{\varphi}((s,\tau))) $, for all $\tau\in[t+t_1,t+t_2]$. Applying these error bounds in Lm \ref{Lm:Fundamentalinequality}, we get
\begin{equation}\label{Eq:ApproxErrorBandSemantic5Step1}
\begin{aligned}
\max_{\tau\in[t+t_1,t+t_2]}\{\rho^\varphi((s,\tau))\} - \max_{\tau\in[t+t_1,t+t_2]}\{\tilde{\rho}^\varphi((s,\tau))\}\\
\in [ \min_{\tau\in[t+t_1,t+t_2]}\{L_\theta^{\varphi}((s,\tau))\},\ \max_{\tau\in[t+t_1,t+t_2]}\{U_\theta^{\varphi}((s,\tau))\}].
\end{aligned}
\end{equation}

Further, using \eqref{Eq:ApproxErrorBandsOfSmoothOperators}, we can write 
\begin{equation}\label{Eq:ApproxErrorBandSemantic5Step2}
\begin{aligned}
    \max_{\tau\in[t+t_1,t+t_2]}\{\tilde{\rho}^\varphi((s,\tau))\} - \underset{\tau\in[t+t_1,t+t_2]}{\widetilde{\max}}\{\tilde{\rho}^\varphi((s,\tau))\} \\ 
    \in [L_{k_2,t_2-t_1+1}^{\max},U_{k_2,t_2-t_1+1}^{\max}]
\end{aligned}
\end{equation}

Finally, by adding \eqref{Eq:ApproxErrorBandSemantic5Step1} and \eqref{Eq:ApproxErrorBandSemantic5Step2} we can obtain a bound for the approximation error term $\tilde{e}^{\mathbf{F}_{[t_1,t_2]}\varphi}((s,t))$ as    
\begin{equation}
\begin{aligned}
\tilde{e}^{\mathbf{F}_{[t_1,t_2]}\varphi}((s,t)) 
&\equiv
[L_{\theta'}^{\mathbf{F}_{[t_1,t_2]}\varphi}((s,t)),\ U_{\theta'}^{\mathbf{F}_{[t_1,t_2]}\varphi}((s,t))],\\
&\in [ 
    L_{k_2,t_2-t_1+1}^{\max}+\min_{\tau\in[t+t_1,t+t_2]}\{L_\theta^{\varphi}((s,\tau))\},\\ 
&\ \ \ \ \ 
U_{k_2,t_2-t_1+1}^{\max}+\max_{\tau\in[t+t_1,t+t_2]}\{U_\theta^{\varphi}((s,\tau))\} ]
\end{aligned}
\end{equation}
where $\theta' = \theta \cup \{k_2\}$. This proves the Semantic 5. 

\paragraph*{Semantic 6} Can be proved by combining the techniques used to prove Semantics 3, 5 and 6. Due to space constraints, details are omitted here. This completes the proof of Th. \ref{Th:STLErrorSemantics}.
\end{proof}

\begin{remark}
Note that we are not constrained to use the same $k_1, k_2$ parameter values across all the smooth-min, smooth-max operators required when defining an SRM for a specification. In other words, at different levels of the specification, different $k_1, k_2$ values can be used. Having this flexibility is especially useful when the components of the specification are not properly normalized. 
One such example, identified using the proposed STL error semantics, is discussed in Section \ref{Sec:SimulationResults}. 
\end{remark}

To conclude this section, we propose the following set of semantics (we named the \emph{Supplementary STL error semantics}) applicable for a situation: 1) where there are no mission space uncertainties (i.e., when $L^{\pi}(s_t) = U^\pi(s_t) = 0,  \forall s_t, \forall \pi$ in \eqref{Eq:ModifiedSmoothSTLRobustSemantic}), or, 2) where mission space uncertainties are independent of the agent trajectory $s$ (i.e., when $L^\pi(s_t) = L^\pi$ and $U^\pi(s_t) = U^\pi$ where $L^\pi$ and $U^\pi$ are known constants). As pointed out in Remark  \ref{Rm:NoMissionSpaceUncertainities}, in such scenarios, the resulting error bounds will be independent of the signal $s$. This property is more evident in the semantics defined below (compared to the semantics defined in Def. \ref{Def:STLErrorSemantics}). 

\begin{definition}\label{Def:SuppSmoothSTLRobustSemantics}
(Supplementary STL Error Semantics)
\begin{itemize}
\item  
$\tilde{e}^\pi \in [-U^\pi,\ -L^\pi]$
\item 
$\tilde{e}^{\neg \pi} \in [L^\pi,\ U^\pi]$
\item 
$\tilde{e}^{\varphi_1 \land \varphi_2} \in 
[L_{k_1,2}^{\min}+\min\{L_{\theta_1}^{\varphi_1},L_{\theta_2}^{\varphi_2}\},\ U_{k_1,2}^{\min}$
\item[] 
\hfill $+\max\{U_{\theta_1}^{\varphi_1},U_{\theta_2}^{\varphi_2}\}]$
\item 
$\tilde{e}^{\varphi_1 \lor \varphi_2} \in [ 
L_{k_2,2}^{\max}+\min\{L_{\theta_1}^{\varphi_1},L_{\theta_2}^{\varphi_2}\},\ U_{k_2,2}^{\max}$
\item[] 
\hfill $+\max\{U_{\theta_1}^{\varphi_1},U_{\theta_2}^{\varphi_2}\}]$
\item
$\tilde{e}^{\mathbf{F}_{[t_1,t_2]}\varphi} \in [L_{k_2,t_2-t_1+1}^{\max} + L_\theta^{\varphi},\ U_{k_2,t_2-t_1+1}^{\max}+U_\theta^{\varphi}]$ 
\item 
$\tilde{e}^{\mathbf{G}_{[t_1,t_2]}\varphi} 
\in [L_{k_1,t_2-t_1+1}^{\min}+L_\theta^{\varphi},\ 
U_{k_1,t_2-t_1+1}^{\min}+U_\theta^{\varphi}]$
\item 
$\tilde{e}^{\varphi_1\mathbf{U}_{[t_1,t_2]}\varphi_2} 
\in [L_{k_2,t_2-t_1+1}^{\max} + L_{\theta_4}^{\varphi_4},\ 
U_{k_2,t_2-t_1+1}^{\max} + U_{\theta_4}^{\varphi_4}]
$
\item[] where
\item[]
$[L_{\theta_4}^{\varphi_4}, U_{\theta_4}^{\varphi_4}] \equiv\ 
[L_{k_1,2}^{\min} + \min\{L_{\theta_1}^{\varphi_1},
L_{\theta_3}^{\varphi_3}\},\ U_{k_1,2}^{\min}$
\item[]
\hfill $+ \max\{U_{\theta_1}^{\varphi_1},U_{\theta_3}^{\varphi_3}\}],$
\item[]
$[L_{\theta_3}^{\varphi_3},U_{\theta_3}^{\varphi_3}]\equiv\ 
[L_{k_1,\tau-t-t_1+1} + L_{\theta_2}^{\varphi_2},\ 
U_{k_1,\tau-t-t_1+1} + U_{\theta_2}^{\varphi_2}]$.
\end{itemize}
\end{definition}

\begin{corollary}
For a given specification $\varphi$ and a smooth robustness measure (SRM1-4), the \emph{STL error semantics} in Def. \ref{Def:SuppSmoothSTLRobustSemantics} can be used to determine the approximation error bound 
$$[L_\theta^\varphi, U_\theta^\varphi] \ni \tilde{e}^\varphi((s,t)) \triangleq (\rho^\varphi((s,t))-\tilde{\rho}^\varphi((s,t))),$$ 
(the same as \eqref{Eq:ApproximationErrorBandDefinition}) in terms of the used collection of smooth-min and smooth-max operator parameters $\theta$.
\end{corollary}
\begin{proof}
The first two semantics in Def. \ref{Def:SuppSmoothSTLRobustSemantics} can be obtained by applying the relationships $L^{\pi}(s_t) = L^\pi$ and  $U^\pi(s_t) = U^\pi$ in respective semantics in Def. \ref{Def:STLErrorSemantics}. As a result of these first two semantics being independent of the signal suffix $(s,t)$, all the remaining semantics in Def. \ref{Def:SuppSmoothSTLRobustSemantics} will also be independent of the signal suffice $(s,t)$. These remaining semantics in Def. \ref{Def:SuppSmoothSTLRobustSemantics} can be proved by following the same steps used to establish their respective semantics in Def. \ref{Def:STLErrorSemantics}. 
\end{proof}

\section{The Gradients of Smooth Robustness Measures}
\label{Sec:Gradients}

Recall that, in Section \ref{Sec:ProblemFormulation}, we formulated the interested synthesis problem \eqref{Eq:SynthesisProblem} as an unconstrained optimization problem \eqref{Eq:SynthesisProblem3}. Then, since the involved robustness measure objective $\rho^\varphi(s(u))$ in \eqref{Eq:SynthesisProblem3} is non-smooth, in Section \ref{Sec:SmoothRobustnessMeasures}, we proposed to replace it with a smooth approximation (an SRM) $\tilde{\rho}^\varphi(s(u))$ inspired by the recent literature \cite{Pant2017,Haghighi2019,Mehdipour2019,Gilpin2021}. The main motivation behind this replacement is to enable the use of efficient gradient-based optimization techniques to solve \eqref{Eq:SynthesisProblem3}. Hence, the importance of knowing the gradient of the used SRM $\tilde{\rho}^\varphi(s(u))$ (with respect to $u$) is clear. 
However, to the best of the authors' knowledge, experimental results reported in existing literature \cite{Pant2017,Haghighi2019,Mehdipour2019,Gilpin2021} relies on numerically evaluated gradients rather than explicitly derived gradients (of SRMs). Intuitively, it is worth investigating the accuracy and the efficiency of such numerically evaluated gradients compared to explicitly derived gradients. 
Motivated by this need, in this section, we propose an efficient systematic approach to explicitly derive the gradient of the used SRM $\tilde{\rho}^\varphi(s(u))$, i.e., $\frac{\partial \tilde{\rho}^\varphi(s(u))}{\partial u}$.


First, notice that we can use the chain rule to write 
\begin{equation}\label{Eq:MainChainRule}
    \frac{\partial \tilde{\rho}^\varphi(s(u))}{\partial u} = \frac{\partial \tilde{\rho}^\varphi(s(u))}{\partial s} \frac{\partial s(u)}{\partial u}. 
\end{equation}
At this point, we remind that: 1) the SRM $\tilde{\rho}^\varphi(s(u))\in \R$, 2) the composite signal $s(u) \triangleq \{s_t(u)\}_{t\in[0,T]}$ where $s_t(u) \triangleq [y_t^\top(u),x_t^\top(u),u_t^\top]^\top \in \R^{q}$ with $q \triangleq p+n+m$, and 3) the control $u = \{u_t\}_{t\in[0,T]}$ where $u_t\in\R^m$ (all vectors are column vectors). For notational simplicity (also without any ambiguity), let us denote the three gradient terms in \eqref{Eq:MainChainRule} as 
$$
\frac{\partial \tilde{\rho}^\varphi}{\partial u} \triangleq \frac{\partial \tilde{\rho}^\varphi(s(u))}{\partial u},\ \ 
\frac{\partial \tilde{\rho}^\varphi}{\partial s} \triangleq \frac{\partial \tilde{\rho}^\varphi(s(u))}{\partial s}\ \mbox{ and }\  
\frac{\partial s}{\partial u} \triangleq 
\frac{\partial s(u)}{\partial u}, 
$$
where $\tilde{\rho}^\varphi \in \R$, $s \in \R^{(T+1)q}$ and $u\in\R^{(T+1)q}$ (consistent with the previously mentioned vector dimensions). Note that, the composition of these three gradient terms are as follows:
\begin{equation}
   \frac{\partial \tilde{\rho}^\varphi}{\partial u} 
   =\left[
   \frac{\partial \tilde{\rho}^\varphi}{\partial u_0}, 
   \frac{\partial \tilde{\rho}^\varphi}{\partial u_1},
   \ldots,
   \frac{\partial \tilde{\rho}^\varphi}{\partial u_T} 
   \right] 
   \in \R^{1 \times (T+1)m}, 
\end{equation}
where each $\frac{\partial \tilde{\rho}^\varphi}{\partial u_t} \in \R^{1 \times m}$ for all $t\in [0,T]$,
\begin{equation}
   \frac{\partial \tilde{\rho}^\varphi}{\partial s} 
   =\left[
   \frac{\partial \tilde{\rho}^\varphi}{\partial s_0}, 
   \frac{\partial \tilde{\rho}^\varphi}{\partial s_1},
   \ldots,
   \frac{\partial \tilde{\rho}^\varphi}{\partial s_T} 
   \right] 
   \in \R^{1 \times (T+1)q}, 
\end{equation}
where each $\frac{\partial \tilde{\rho}^\varphi}{\partial s_t} \in \R^{1 \times q}$ for all $t\in [0,T]$ and
\begin{equation}\label{Eq:GradsvsuBlockForm}
    \frac{\partial s}{\partial u} = 
    \begin{bmatrix}
        \frac{\partial s_0}{\partial u_0} & 
        \frac{\partial s_0}{\partial u_1} & 
        \cdots & 
        \frac{\partial s_0}{\partial u_T} \\
        \frac{\partial s_1}{\partial u_0} & 
        \frac{\partial s_1}{\partial u_1} & 
        \cdots & 
        \frac{\partial s_1}{\partial u_T} \\
        \vdots & 
        \vdots & 
        \ddots &
        \vdots \\
        \frac{\partial s_T}{\partial u_0} & 
        \frac{\partial s_T}{\partial u_1} & 
        \cdots & 
        \frac{\partial s_T}{\partial u_T}
    \end{bmatrix}
    \in \R^{(T+1)q \times (T+1)m},
\end{equation}
where each $\frac{\partial s_t}{\partial u_\tau} \in \R^{q \times m}$ for all $t,\tau \in [0,T]$. 

To obtain the interested gradient term $\frac{\partial \tilde{\rho}^\varphi}{\partial u}$, we need to compute the gradient terms $\frac{\partial \tilde{\rho}^\varphi}{\partial s}$ and $\frac{\partial s}{\partial u}$. For this purpose, in the following two subsections, we respectively propose a novel set of semantics and exploit the agent dynamics \eqref{Eq:AgentDynamics}.   

\subsection{The gradient term: $\frac{\partial \tilde{\rho}^\varphi(s)}{\partial s}$}


In the following analysis, we use the notation $0_k$ to represent the zero vector in $\R^k$. Note also that, as shown in Tab. \ref{Tab:SmoothOperators}, we already have derived the gradients of smooth-min and smooth-max operators (i.e., $\frac{\partial \widetilde{\min}\{a\}}{\partial a_i}$, $\frac{\partial \widetilde{\max}\{a\}}{\partial a_i} \in \R$, for all $i$) used in different SRMs: SRM1-4.  
Now, consider the set of semantics (we named the \emph{STL gradient semantics})  defined below. 

\begin{definition}\label{Def:STLGradientSemanctics}
(STL gradient semantics)
\begin{itemize}\setlength\itemsep{0.5em}
\item
$\frac{\partial \tilde{\rho}^\pi((s,t))}{\partial s} = 
\left[0_{tq}^\top,\, \frac{\partial \mu^\pi(s_t)}{\partial s_t},\,0_{(T-t)q}^\top\right] \in \R^{1\times (T+1)q}$ 
\item
$\frac{\partial \tilde{\rho}^{\neg\pi}((s,t))}{\partial s} = -\frac{\partial \tilde{\rho}^{\pi}((s,t))}{\partial s}$
\item
$\frac{\partial \tilde{\rho}^{\varphi_1 \land \varphi_2}((s,t))}{\partial s} 
= \frac{\partial \widetilde{\min} \{\rho_1,\rho_2\}}{\partial \rho_1} 
\frac{\partial \rho_1}{\partial s} 
+ 
\frac{\partial \widetilde{\min} \{\rho_1,\, \rho_2\}}{\partial \rho_2} 
\frac{\partial \rho_2}{\partial s}$
\item[] where $\rho_1 \triangleq \tilde{\rho}^{\varphi_1}((s,t))$ and $\rho_2 \triangleq  \tilde{\rho}^{\varphi_2}((s,t))$
\item 
$\frac{\partial \tilde{\rho}^{\varphi_1 \lor  \varphi_2}((s,t))}{\partial s} 
= \frac{\partial \widetilde{\max} \{\rho_1,\rho_2\}}{\partial \rho_1} 
\frac{\partial \rho_1}{\partial s} 
+ 
\frac{\partial \widetilde{\max} \{\rho_1,\, \rho_2\}}{\partial \rho_2} 
\frac{\partial \rho_2}{\partial s}$
\item[] where $\rho_1 \triangleq \tilde{\rho}^{\varphi_1}((s,t))$ and $\rho_2 \triangleq  \tilde{\rho}^{\varphi_2}((s,t))$
\item 
$\frac{\partial \tilde{\rho}^{\mathbf{F}_{[t_1,t_2]}\varphi}((s,t))}{\partial s} = \sum_{\tau\in[t+t_1,t+t_2]} \frac{\partial \widetilde{\max}_{\tau\in[t+t_1,t+t_2]}\{\rho_\tau\}}{\partial \rho_\tau} \frac{\partial \rho_\tau}{\partial s}$
\item[] where $\rho_\tau \triangleq \tilde{\rho}^\varphi((s,\tau))$ for any $\tau\in[t+t_1,t+t_2]$
\item 
$\frac{\partial \tilde{\rho}^{\mathbf{G}_{[t_1,t_2]}\varphi}((s,t))}{\partial s} = \sum_{\tau\in[t+t_1,t+t_2]} \frac{\partial \widetilde{\min}_{\tau\in[t+t_1,t+t_2]}\{\rho_\tau\}}{\partial \rho_\tau} \frac{\partial \rho_\tau}{\partial s}$
\item[] where $\rho_\tau \triangleq \tilde{\rho}^\varphi((s,\tau))$ for any $\tau\in[t+t_1,t+t_2]$
\item 
$\frac{\partial \tilde{\rho}^{\varphi_1\mathbf{U}_{[t_1,t_2]}\varphi_2}((s,t))}{\partial s} = 
\sum_{\tau \in [t+t_1,t+t_2]} 
\frac{\partial \widetilde{\max}_{\tau\in[t+t_1,t+t_2]}\{\rho_\tau\}}{\partial \rho_\tau} 
\frac{\partial \rho_\tau}{\partial s}
$,
\item[]
where 
$
\frac{\partial \rho_\tau}{\partial s} = 
    \frac{\partial \widetilde{\min}\{\rho_{1,\tau}, \rho_{2,\tau}\}}{\rho_{1,\tau}}
    \frac{\partial \rho_{1,\tau}}{\partial s} + 
    \frac{\partial \widetilde{\min}\{\rho_{1,\tau}, \rho_{2,\tau}\}}{\rho_{2,\tau}}
    \frac{\partial \rho_{2,\tau}}{\partial s} 
$ 
\item[] and
$
\frac{\partial \rho_{2,\tau}}{\partial s} = 
    \sum_{\delta \in [t+t_1,\tau]} \frac{\partial \widetilde{\min}_{\delta\in[t+t_1,\tau]}\{\rho_{2,\delta}\}}{\partial \rho_{2,\delta}} 
    \frac{\partial \rho_{2,\delta}}{\partial s},
$ 
\item[] with 
$\rho_\tau \triangleq \widetilde{\min}\{\rho_{1,\tau}, \rho_{2,\tau}\}$, 
$\rho_{1,\tau} \triangleq \tilde{\rho}^{\varphi_1}((s,\tau))$, 
\item[]
$\rho_{2,\tau} \triangleq \widetilde{\min}_{\delta\in[t+t_1,\tau]}\{\rho_{2,\delta}\}$ and
$\rho_{2,\delta} \triangleq \tilde{\rho}^{\varphi_2}((s,\delta))$.
\end{itemize}
\end{definition}

\begin{theorem}\label{Th:STLGradientSemantics}
For a known signal suffix $(s,t)$, a specification $\varphi$
and a smooth robustness measure (SRM1-4), the \emph{STL gradient semantics} in Def. \ref{Def:STLGradientSemanctics} can be used to determine the gradient 
$
\frac{\partial \tilde{\rho}^\varphi((s,t))}{\partial s}
$ (in the form: 
$
\frac{\partial \tilde{\rho}^\varphi((s,t))}{\partial s} = 
[0_{tq}^\top, 
\frac{\partial \tilde{\rho}^\varphi((s,t))}{\partial s_t},
\ldots, 
\frac{\partial \tilde{\rho}^\varphi((s,t))}{\partial s_T}
] \in \R^{1\times (T+1)q}
$).
\end{theorem}
\begin{proof}
In Semantics 3-7 of Def. \ref{Def:STLGradientSemanctics}, note that all the involved gradient terms of smooth operator values are scalars (e.g., $\frac{\partial \widetilde{\min} \{\rho_1,\rho_2\}}{\partial \rho_1},\, \frac{\partial \widetilde{\min} \{\rho_1,\rho_2\}}{\partial \rho_2} \in \R$ in Semantic 3, see also Tab. \ref{Tab:SmoothOperators}). Hence, in those semantics, such scalar gradients act only as scaling factors on the respective vector-valued \emph{constituent gradients} (e.g., on $\frac{\partial \rho_1}{\partial s},\, \frac{\partial \rho_2}{\partial s} \in \R^{1\times (T+1)q}$ in Semantic 3). 
Therefore, the dimension of the resulting gradient upon using a Semantic 3-7 will be the same as that of the corresponding constituent gradients used. 
Now, since each such constituent gradient has its root in Semantic 1 - that always assigns a vector in $\R^{1\times(T+1)q}$ as the gradient value - it is clear that the semantics in Def. \ref{Def:STLGradientSemanctics} will always produce a vector in $\R^{1\times (T+1)q}$ as the gradient $\frac{\partial \tilde{\rho}^\varphi((s,t))}{\partial s}$.

In fact, each of the Semantics 3-7 in Def. \ref{Def:STLGradientSemanctics} can be established easily by differentiating the respective semantic in Def.  \ref{Def:SmoothSTLRobustSemantics} (with respect to $s$) using the chain rule. For example, the Semantic 3 in Def. \ref{Def:STLGradientSemanctics} is proved by differentiating the Semantic 3 in Def. \ref{Def:SmoothSTLRobustSemantics} as
\begin{align*}
\frac{\partial \tilde{\rho}^{\varphi_1 \land \varphi_2}((s,t))}{\partial s} 
=& 
\frac{\partial \widetilde{\min}\{\tilde{\rho}^{\varphi_1}((s,t)), \tilde{\rho}^{\varphi_2}((s,t))\} }{\partial s} \\
=& \frac{\partial \widetilde{\min} \{\rho_1,\rho_2\}}{\partial \rho_1} 
\frac{\partial \rho_1}{\partial s} + 
\frac{\partial \widetilde{\min} \{\rho_1,\, \rho_2\}}{\partial \rho_2} 
\frac{\partial \rho_2}{\partial s}
\end{align*}
where $\rho_1 \triangleq \tilde{\rho}^{\varphi_1}((s,t))$ and $\rho_2 \triangleq \tilde{\rho}^{\varphi_2}((s,t))$. 
Similarly, even without applying the chain rule, the Semantics 1-2 in Def. \ref{Def:STLGradientSemanctics} can be proved by simply differentiating the respective semantics in Def.  \ref{Def:SmoothSTLRobustSemantics}. Therefore, in parallel to the way that the semantics in Def. \ref{Def:SmoothSTLRobustSemantics} computes the SRM $\tilde{\rho}^\varphi((s,t))$, the semantics in Def. \ref{Def:STLGradientSemanctics} computes the gradient of the SRM: $\frac{\partial \tilde{\rho}^\varphi((s,t))}{\partial s}$. Finally, note that, by definition, $\tilde{\rho}^\varphi((s,t)) = \tilde{\rho}^\varphi(\{s_t,s_{t+1},\ldots,s_T\})$ and $s=(s,0)=\{s_0,s_1,\ldots,s_T\}$. Hence, $\tilde{\rho}^\varphi((s,t))$ is independent of the signal values $\{s_0,s_1,\ldots,s_{t-1}\}$, and thus, the first $tq$ terms in $\frac{\partial \tilde{\rho}^\varphi((s,t))}{\partial s}$ will always be zeros. This completes the proof. 
\end{proof}

Finally, we point out that $\frac{\partial \tilde{\rho}^\varphi(s)}{\partial s} = \frac{\partial \tilde{\rho}^\varphi((s,0))}{\partial s}$ as $s = (s,0)$. Hence, we now can use the STL gradient semantics proposed in Def. \ref{Def:STLGradientSemanctics} (proved in Th. \ref{Th:STLErrorSemantics}) to compute the gradient term $\frac{\partial \tilde{\rho}^\varphi(s)}{\partial s}$ required in \eqref{Eq:MainChainRule}. In the next section, we will derive the remaining gradient term required in \eqref{Eq:MainChainRule}, i.e., $\frac{\partial s(u)}{\partial u}$.

\subsection{The gradient term: $\frac{\partial s(u)}{\partial u}$}

In the following analysis, with regard to the agent dynamics model \eqref{Eq:AgentDynamics}, we use the notation $\frac{\partial f}{\partial x}\big\vert_t$ to represent the  $\frac{\partial f(x,u)}{\partial x}$ evaluated at $(x_t,u_t)$ where $t\in[0,T]$. The same convention applies for the notations $\frac{\partial f}{\partial u}\big\vert_t$, 
$\frac{\partial g}{\partial x}\big\vert_t$ and   
$\frac{\partial g}{\partial u}\big\vert_t$. Since the functions $f(x,u)$ and $g(x,u)$ in \eqref{Eq:AgentDynamics} are assumed as a given and differentiable, their partial derivatives, i.e., $\frac{\partial f(x,u)}{\partial x}$, $\frac{\partial f(x,u)}{\partial u}$, $\frac{\partial g(x,u)}{\partial x}$ and $\frac{\partial g(x,u)}{\partial u}$ are also considered as a given. 
Therefore, for a given control signal $u = \{u_t\}_{t\in[0,T]}$ (which also determines the generated state signal $x= \{x_t\}_{t\in[0,T]}$ through \eqref{Eq:AgentDynamics}),  terms like $\frac{\partial f}{\partial x}\big\vert_t$, 
$\frac{\partial f}{\partial u}\big\vert_t$, 
$\frac{\partial g}{\partial x}\big\vert_t$ and   
$\frac{\partial g}{\partial u}\big\vert_t$ can be evaluated efficiently for any $t\in[0,T]$. 

In order to determine the the required gradient term: $\frac{\partial s}{\partial u} = \frac{\partial s(u)}{\partial u}$, we first need to establish the following two lemmas.

\begin{lemma}\label{Lm:Gradyvsu}
Under the agent dynamics model \eqref{Eq:AgentDynamics}, the control input signal $u = \{u_t\}_{t\in[0,T]}$ and the output signal $y=\{y_t\}_{t\in[0,T]}$ are related such that, for all  $t,\tau \in [0,T]$,
\begin{equation}\label{Eq:Lm:Gradyvsu}
    \frac{\partial y_t}{\partial u_\tau} = 
    \begin{cases}
    \frac{\partial g}{\partial x}\big\vert_t  
    \prod_{i=1}^{t-\tau-1} \left[\frac{\partial f}{\partial x}\Big\vert_{t-i} \right]
    \frac{\partial f}{\partial u}\Big\vert_{\tau}, \ \ &\mbox{if } \tau < t,\\
    \frac{\partial g}{\partial u}\big\vert_t, \ \ &\mbox{if } \tau = t,\\
    0, \ \ &\mbox{if } \tau > t.\\ 
    \end{cases}
\end{equation}
\end{lemma}

\begin{proof}
According to the agent dynamics model \eqref{Eq:AgentDynamics}, 
\begin{equation}\label{Eq:AgentOutputRecursive}
\begin{aligned}
y_t 
=& g(x_t,u_t) \\
=& g(f(x_{t-1},u_{t-1}),u_t)\\
=& g(f(f(x_{t-2},u_{t-2}),u_{t-1}),u_t)\\
\vdots\,& 
\end{aligned}
\end{equation}
Therefore, $\frac{\partial y_t}{\partial u_\tau} = 0$ for any $\tau>t$ and $\frac{\partial y_t}{\partial u_t} = \frac{\partial g}{\partial u}\big\vert_t$. Similarly, using \eqref{Eq:AgentOutputRecursive} together with the product rule, we can write
\begin{align*}
  \frac{\partial y_t}{\partial u_{t-1}} =& \frac{\partial g}{\partial x}\Big\vert_t \frac{\partial f}{\partial u}\Big\vert_{t-1}\\ 
\frac{\partial y_t}{\partial u_{t-2}} =& \frac{\partial g}{\partial x}\Big\vert_t \frac{\partial f}{\partial x}\Big\vert_{t-1} \frac{\partial f}{\partial u}\Big\vert_{t-2}  \\
\frac{\partial y_t}{\partial u_{t-3}} =& \frac{\partial g}{\partial x}\Big\vert_t \frac{\partial f}{\partial x}\Big\vert_{t-1} \frac{\partial f}{\partial x}\Big\vert_{t-2} \frac{\partial f}{\partial u}\Big\vert_{t-3}  \\
\vdots\ &
\end{align*}
Therefore, it is clear that \eqref{Eq:Lm:Gradyvsu} provides an accurate and concise representation of these results. 
\end{proof}

\begin{lemma}\label{Lm:Gradxvsu}
Under the agent dynamics model \eqref{Eq:AgentDynamics}, the control input signal $u = \{u_t\}_{t\in[0,T]}$ and the state signal $x=\{x_t\}_{t\in[0,T]}$ are related such that, for all  $t,\tau \in [0,T]$,
\begin{equation}\label{Eq:Lm:Gradxvsu}
    \frac{\partial x_t}{\partial u_\tau} = 
    \begin{cases}
    \prod_{i=1}^{t-\tau-1} \left[\frac{\partial f}{\partial x}\Big\vert_{t-i} \right]
    \frac{\partial f}{\partial u}\Big\vert_{\tau}, \ \ &\mbox{if } \tau < t,\\
    0, \ \ &\mbox{if } \tau \geq t.\\ 
    \end{cases}
\end{equation}
\end{lemma}
\begin{proof}
From the agent dynamics model \eqref{Eq:AgentDynamics}, we can write  
\begin{equation}\label{Eq:AgentStateRecursive}
\begin{aligned}
x_t =& f(x_{t-1},u_{t-1}) \\
=& f(f(x_{t-2},u_{t-2}),u_{t-1}) \\
=& f(f(f(x_{t-3,u_{t-3}}),u_{t-2}),u_{t-1}) \\
\vdots\, &
\end{aligned}
\end{equation}
Hence, $\frac{\partial x_t}{\partial u_\tau} = 0$ for any $\tau \geq t$ and $\frac{\partial x_t}{\partial u_{t-1}} = \frac{\partial f}{\partial u}\big\vert_{t-1}$. Similarly, using \eqref{Eq:AgentStateRecursive} together with the product rule, we can write 
\begin{align*}
    \frac{\partial x_t}{\partial u_{t-2}}  =& \frac{\partial f}{\partial x}\Big\vert_{t-1} \frac{\partial f}{\partial u}\Big\vert_{t-2}\\
    \frac{\partial x_t}{\partial u_{t-2}}  =& \frac{\partial f}{\partial x}\Big\vert_{t-1} \frac{\partial f}{\partial x}\Big\vert_{t-2} \frac{\partial f}{\partial u}\Big\vert_{t-3} \\
    \vdots\ &
\end{align*}
Therefore, it is clear that \eqref{Eq:Lm:Gradxvsu} provides an accurate and concise representation of these results. 
\end{proof}

Using these two lemmas, we now propose a theorem that can be used to determine the required gradient term: $\frac{\partial s}{\partial u}$. 

\begin{theorem}\label{Th:STLGradiensDynamics}
Under the agent dynamics model \eqref{Eq:AgentDynamics}, the gradient of the composite signal $s(u)$ with respect to the control signal $u$ is given by a $(T+1)\times(T+1)$ block matrix $\frac{\partial s}{\partial u}$ (as in \eqref{Eq:GradsvsuBlockForm}) with its $(t,\tau)$\textsuperscript{th} block being 
\begin{equation}
    \frac{\partial s_t}{\partial u_\tau} = 
    \begin{bmatrix}
    \frac{\partial y_t}{\partial u_\tau}\\    
    \frac{\partial x_t}{\partial u_\tau}\\    
    I_m
    \end{bmatrix} \in \R^{q \times m}
\end{equation}
where $\frac{\partial y_t}{\partial u_\tau} \in \R^{p \times m}$ is given by \eqref{Eq:Lm:Gradyvsu}, $\frac{\partial x_t}{\partial u_\tau} \in \R^{n \times m}$ is given by \eqref{Eq:Lm:Gradxvsu} and $I_m$ is the identity matrix in $\R^{m \times m}$.    
\end{theorem}
\begin{proof} 
This result directly follows from Lms. \ref{Lm:Gradyvsu}, \ref{Lm:Gradxvsu} and the format of the composite signal $s(u) = \{s_t(u)\}_{t\in[0,T]} \in \R^{1\times(T+1)q}$ where $s_t(u) \triangleq [y_t^\top(u),x_t^\top(u),u_t^\top]^\top$.
\end{proof}

With that, we now can explicitly evaluate the gradient of the used SRM: $\frac{\partial \tilde{\rho}^{\varphi}(s(u))}{\partial u}$ (i.e., the gradient of the objective function used in the synthesis problem \eqref{Eq:SynthesisProblem3}) using \eqref{Eq:MainChainRule} and Theorems \ref{Th:STLGradientSemantics} and \ref{Th:STLGradiensDynamics}.

\section{Simulation Results}
\label{Sec:SimulationResults}

In this section, to highlight our contributions, we consider the four symbolic control problems (SCPs) shown in Fig. \ref{Fig:ProblemConfigurations}. In each SCP, the objective is to synthesize a control signal for a simple simulated robot that follows 2-D integrator dynamics \cite{Gilpin2021} such that it satisfies a given specification (indicated in Fig. \ref{Fig:ProblemConfigurations}). We have implemented: 1) the SCPs shown in Fig. \ref{Fig:ProblemConfigurations}, 2) the SRMs discussed in Section \ref{Sec:SmoothRobustnessMeasures}, 3) the STL error semantics proposed in Section \ref{Sec:ApproximationErrors} and 4) the STL gradient semantics proposed in \ref{Sec:Gradients}, in a Python environment that has been made available at \url{https://github.com/shiran27/Symbolic-Control-via-STL} for reproduction and reuse purposes. 

In order to solve the aforementioned synthesis problem (i.e., \eqref{Eq:SynthesisProblem3}) in an energy-aware manner, we follow the approach used in \cite{Gilpin2021} where a smooth cost function 
\begin{equation}\label{Eq:SmoothCostFunction}
\tilde{J}(u) \triangleq \tilde{\rho}^\varphi(s(u))-0.01\Vert u \Vert^2    
\end{equation}
is optimized (as an alternative to optimizing its non-smooth version: $J(u) \triangleq \rho^\varphi(s(u))-0.01\Vert u \Vert^2$, see also \eqref{Eq:SynthesisProblem4}), using the SciPy's SQP method \cite{Virtanen2020}.

\begin{figure*}[!t]
    \centering
    \begin{subfigure}[t]{0.24\textwidth}
        \centering
        \captionsetup{justification=centering}
        \includegraphics[width=\textwidth]{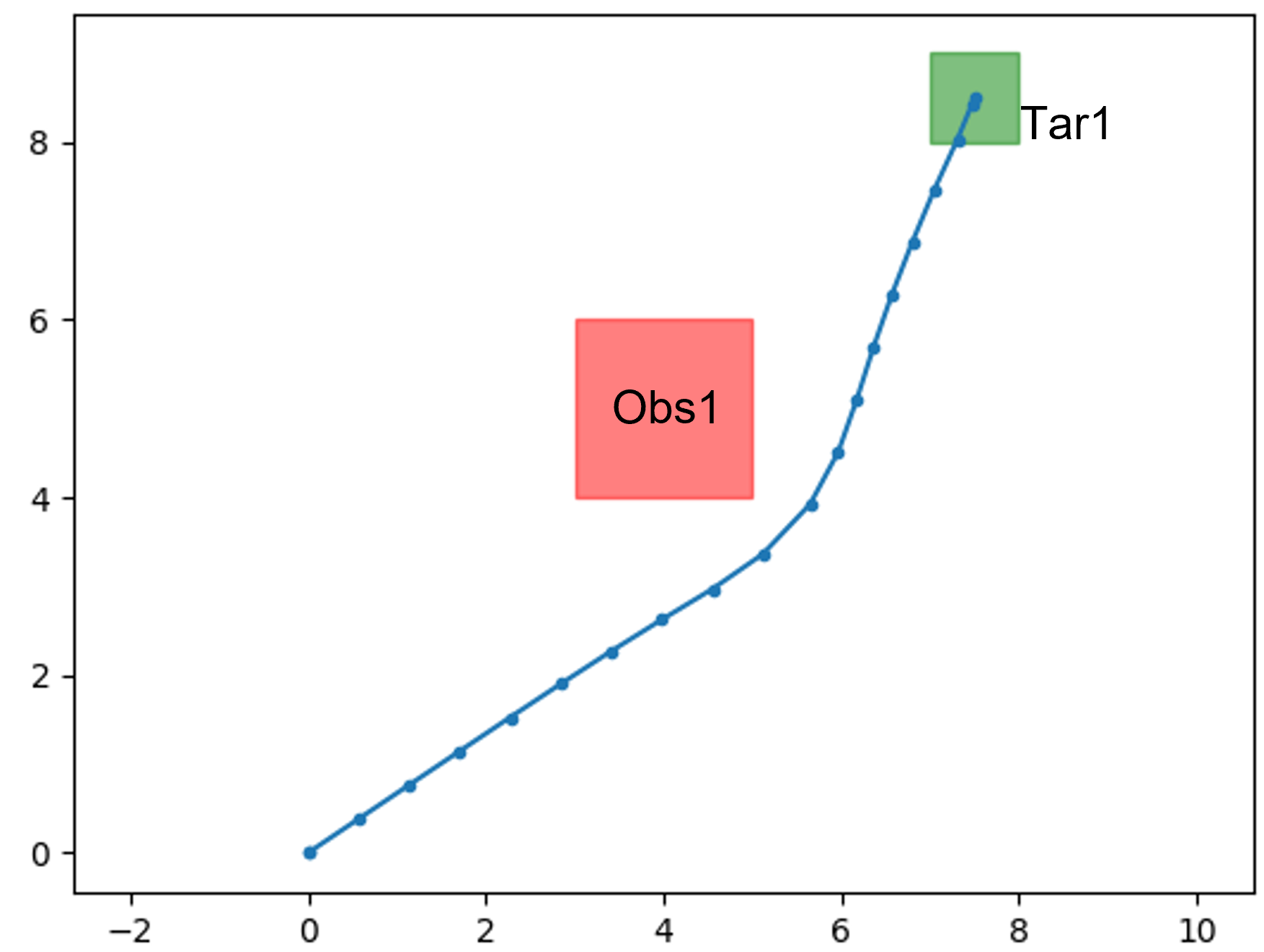}
        \caption{SCP1: The specification is $\varphi = \varphi_o\land\varphi_u\land\textbf{F}_{[0,20]}\mbox{Tar}1$.}
    \end{subfigure}
    \hfill
    \begin{subfigure}[t]{0.24\textwidth}
        \centering
        \captionsetup{justification=centering}
        \includegraphics[width=\textwidth]{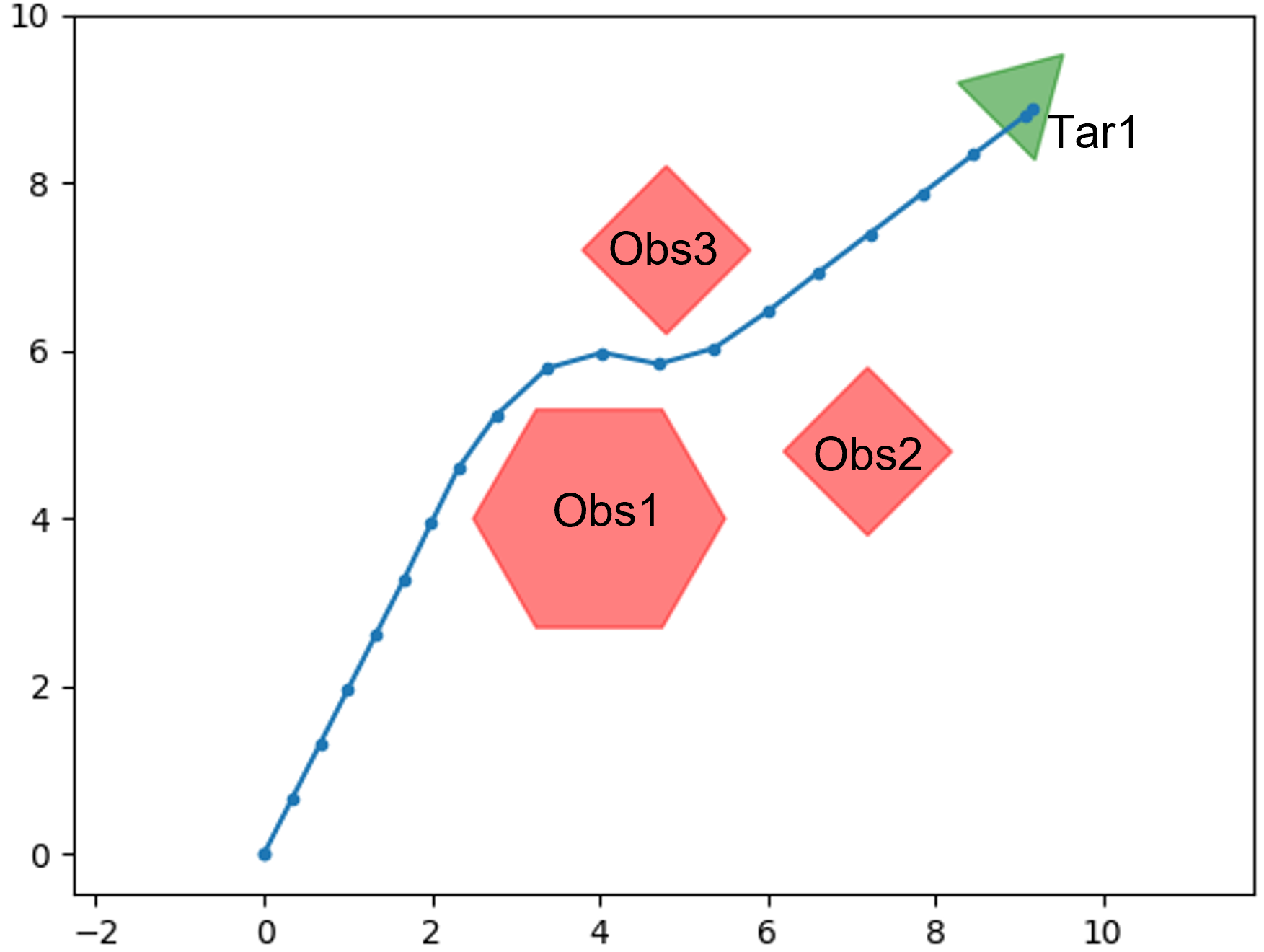}
        \caption{SCP2: The specification is $\varphi = \varphi_o\land\varphi_u\land\textbf{F}_{[0,20]}\mbox{Tar}1$.}
    \end{subfigure}
    \hfill
    \begin{subfigure}[t]{0.24\textwidth}
        \centering
        \captionsetup{justification=centering}
        \includegraphics[width=\textwidth]{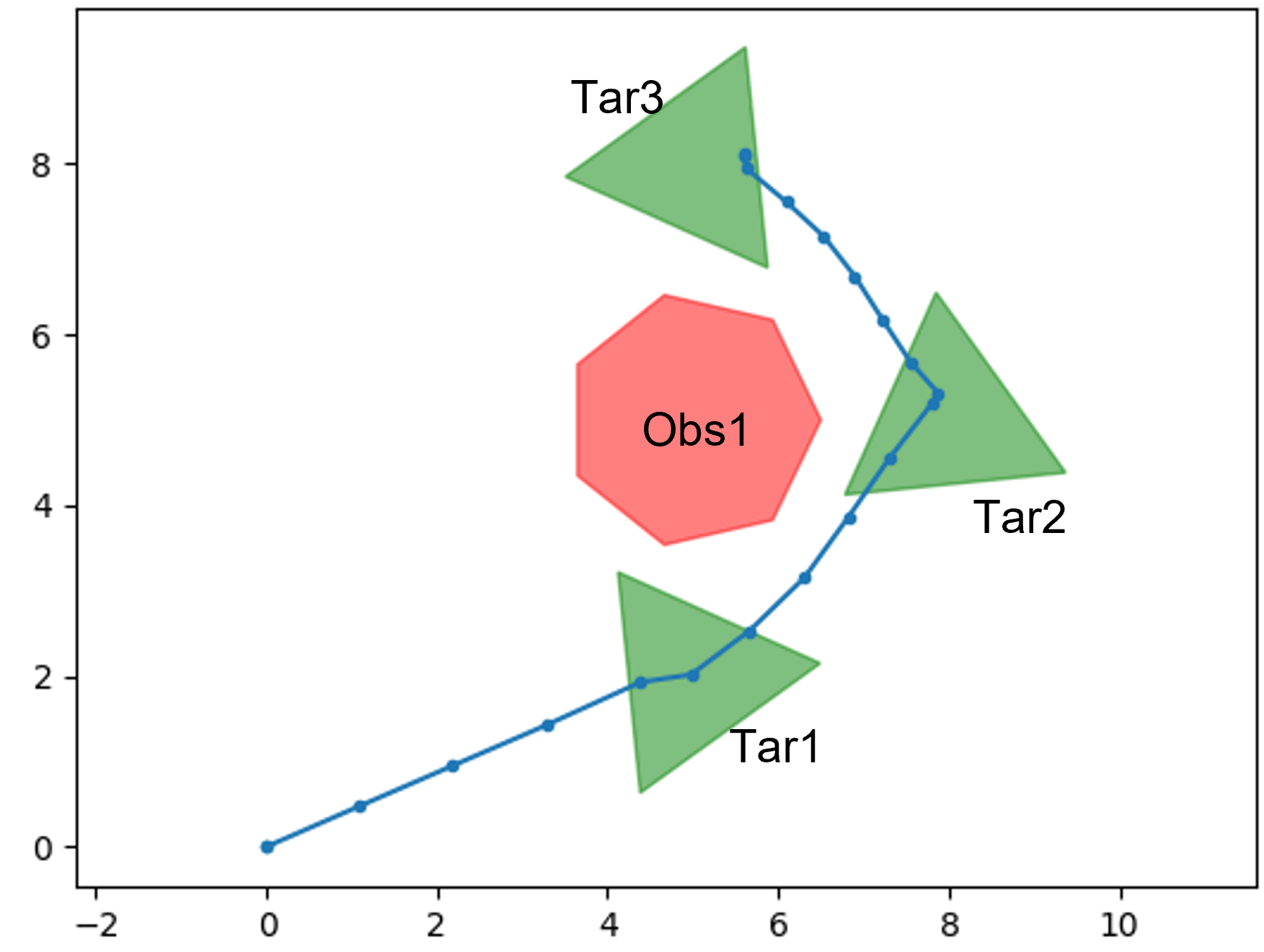}
        \caption{SCP3: The specification is $\varphi = \varphi_o\land\varphi_u\land\textbf{F}_{[0,6]}(\mbox{Tar}1) \land \textbf{F}_{[6,12]}(\mbox{Tar}2) \land \textbf{F}_{[14,20]}\mbox{Tar}3$.}
    \end{subfigure}
    \hfill
    \begin{subfigure}[t]{0.24\textwidth}
        \centering
        \captionsetup{justification=centering}
        \includegraphics[width=\textwidth]{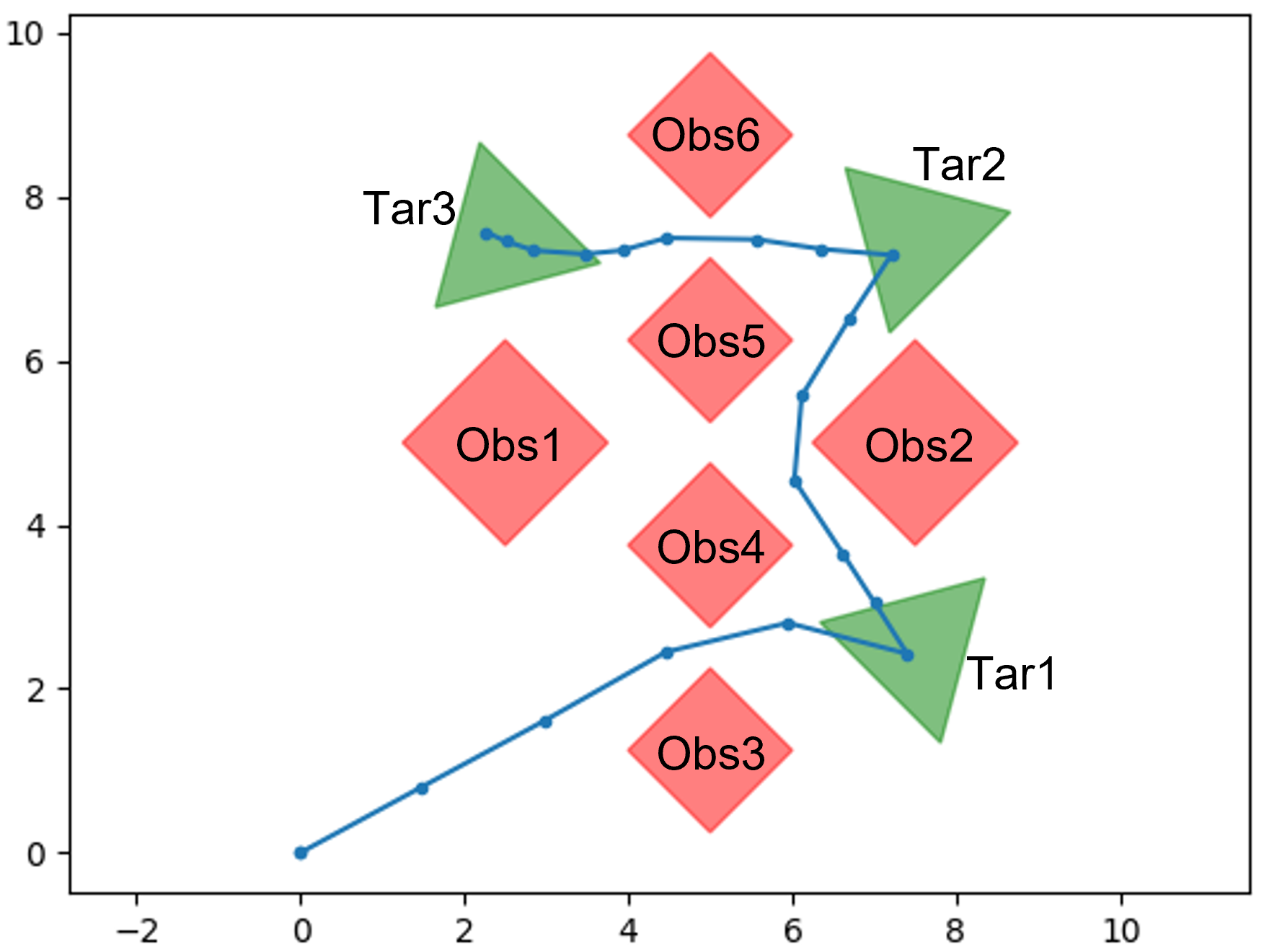}
        \caption{SCP4: The specification is $\varphi = \varphi_o\land\varphi_u\land\textbf{F}_{[0,6]}(\mbox{Tar}1) \land \textbf{F}_{[6,12]}(\mbox{Tar}2) \land \textbf{F}_{[12,20]}\mbox{Tar}3$}
    \end{subfigure}
    \caption{Considered symbolic control problems (SCPs). Note that $\varphi_{o} = \textbf{G}_{[0,20]}\land_{i}(\neg \mbox{Obs}i)$ and $\varphi_u = \textbf{G}_{[0,20]}(-1 \leq u_t \leq 1)$. All agent trajectories have been generated by optimizing \eqref{Eq:SmoothCostFunction} with the SRM3, under smooth operator parameters $k_1=k_2=3$.}
    \label{Fig:ProblemConfigurations}
\end{figure*}

\subsection{Effectiveness of the explicit gradients}

We start with illustrating the efficiency and the accuracy of the explicit gradient computation technique proposed in Section \ref{Sec:Gradients}. For this purpose, we compare such \emph{explicit} gradients with the gradients provided by the symbolic \emph{AutoGrad} software package \cite{Maclaurin2015} (that was used in \cite{Gilpin2021}) - in terms of the execution times and the absolute component differences (errors). It is also worth noting that the said AutoGrad based approach is much more efficient and accurate compared to numerical methods where the gradient is computed using finite differences \cite{Pant2017}.

In our experiments, for each SCP shown in Fig. \ref{Fig:ProblemConfigurations} and for each SRM, the explicit gradient and the AutoGrad based gradient (of the cost function) were evaluated at 500 randomly selected control signals. The observed mean execution times and the ($\log$) mean absolute component errors are reported in Tab. \ref{Tab:GradientResults1}.  According to the observed execution times, the proposed explicit gradient evaluation technique is significantly efficient compared to the AutoGrad based technique as it has improved the mean execution time on average by 57.7\% irrespective of the considered SCP or the SRM. Moreover, the observed error values (in log scale) reveal that both gradient computing techniques essentially provide the same gradient values.   

However, notice a considerable difference in the scale (albeit small in the magnitude) in the observed error values across different SRMs. Therefore, this behavior was further investigated under varying smooth operator parameters $k_1,k_2$, for the SCP3. The observed results are illustrated in Fig. \ref{Fig:GradientResults1}. Intuitively, with increasing $k_1$, $k_2$ values, an SRM becomes approximately close to the actual non-smooth robustness measure, and thus, the corresponding gradients evaluated using the AutoGrad method can be expected to deviate from the actual explicit gradients. However, according to the observations shown in Fig. \ref{Fig:GradientResults1}, such a behavior can only be seen for SRM2 and SRM4. Moreover, error magnitudes associated with SRM3 and SRM1 are lower (mush significantly in the latter case) than those of SRM2 and SRM4. Hence, we can conclude that SRMs like SRM1 and SRM3 can be used even with the AutoGrad package (to compute gradients) without running into significant numerical inaccuracies irrespective of the used $k_1$, $k_2$ values.

\begin{table}[!t]
\caption{A comparison between explicit gradients and AutoGrad based gradients. Each SRM was defined using the smooth operator parameters $k_1=k_2=3$.}
\resizebox{\columnwidth}{!}{
\begin{tabular}{|c|c|r|r|r|r|r|}
\Xhline{2\arrayrulewidth}
\multirow{2}{*}{SCP} & \multirow{2}{*}{SRM} & \multicolumn{3}{c|}{Mean Execution Time / (ms)} & \multicolumn{2}{l|}{Absolute Component Error} \\ \cline{3-7}  
& &  
  \multicolumn{1}{c|}{\begin{tabular}[c]{@{}c@{}}Explicit\end{tabular}} &
  \multicolumn{1}{c|}{\begin{tabular}[c]{@{}c@{}}AutoGrad\end{tabular}} &
  \multicolumn{1}{c|}{\begin{tabular}[c]{@{}c@{}}\% Impr.\end{tabular}} &
  \multicolumn{1}{c|}{$\log_{10}(\text{Mean})$} &
  \multicolumn{1}{c|}{$\log_{10}(\text{St.D.})$} \\ \Xhline{2\arrayrulewidth}
                        & SRM1 & 23.3 & 54.9  & 57.6 & -16.6 & -16.4 \\ \cline{2-7} 
                        & SRM2 & 23.9 & 56.0  & 57.2 & -16.1 & -15.8 \\ \cline{2-7} 
                        & SRM3 & 24.8 & 56.6  & 56.3 & -16.0 & -15.8 \\ \cline{2-7} 
\multirow{-4}{*}{SCP1} & SRM4 & 25.6 & 57.0  & 55.2 & -15.8 & -15.5 \\ \Xhline{2\arrayrulewidth}
                        & SRM1 & 46.8 & 115.0 & 59.3 & -16.2 & -15.9 \\ \cline{2-7} 
                        & SRM2 & 44.1 & 105.8 & 58.3 & -15.7 & -15.3 \\ \cline{2-7} 
                        & SRM3 & 43.8 & 104.2 & 58.0 & -15.6 & -15.3 \\ \cline{2-7} 
\multirow{-4}{*}{SCP2} & SRM4 & 49.8 & 114.5 & 56.5 & -15.4 & -15.1 \\ \Xhline{2\arrayrulewidth}
                        & SRM1 & 24.0 & 58.9  & 59.3 & -15.5 & -15.1 \\ \cline{2-7} 
                        & SRM2 & 24.9 & 60.9  & 59.0 & -13.9 & -13.5 \\ \cline{2-7} 
                        & SRM3 & 26.3 & 63.0  & 58.3 & -14.2 & -13.7 \\ \cline{2-7} 
\multirow{-4}{*}{SCP3} & SRM4 & 27.3 & 64.9  & 57.9 & -13.8 & -13.5 \\ \Xhline{2\arrayrulewidth}
                        & SRM1 & 62.8 & 150.9 & 58.4 & -15.7 & -15.4 \\ \cline{2-7} 
                        & SRM2 & 66.9 & 156.1 & 57.2 & -14.7 & -14.3 \\ \cline{2-7} 
                        & SRM3 & 64.4 & 150.7 & 57.2 & -14.5 & -14.2 \\ \cline{2-7} 
\multirow{-4}{*}{SCP4} & SRM4 & 72.3 & 168.3 & 57.0 & -14.4 & -14.1 \\ \Xhline{2\arrayrulewidth}
\end{tabular}}
\label{Tab:GradientResults1}
\end{table}

\begin{figure}[!t]
    \centering
    \includegraphics[width=1.7in]{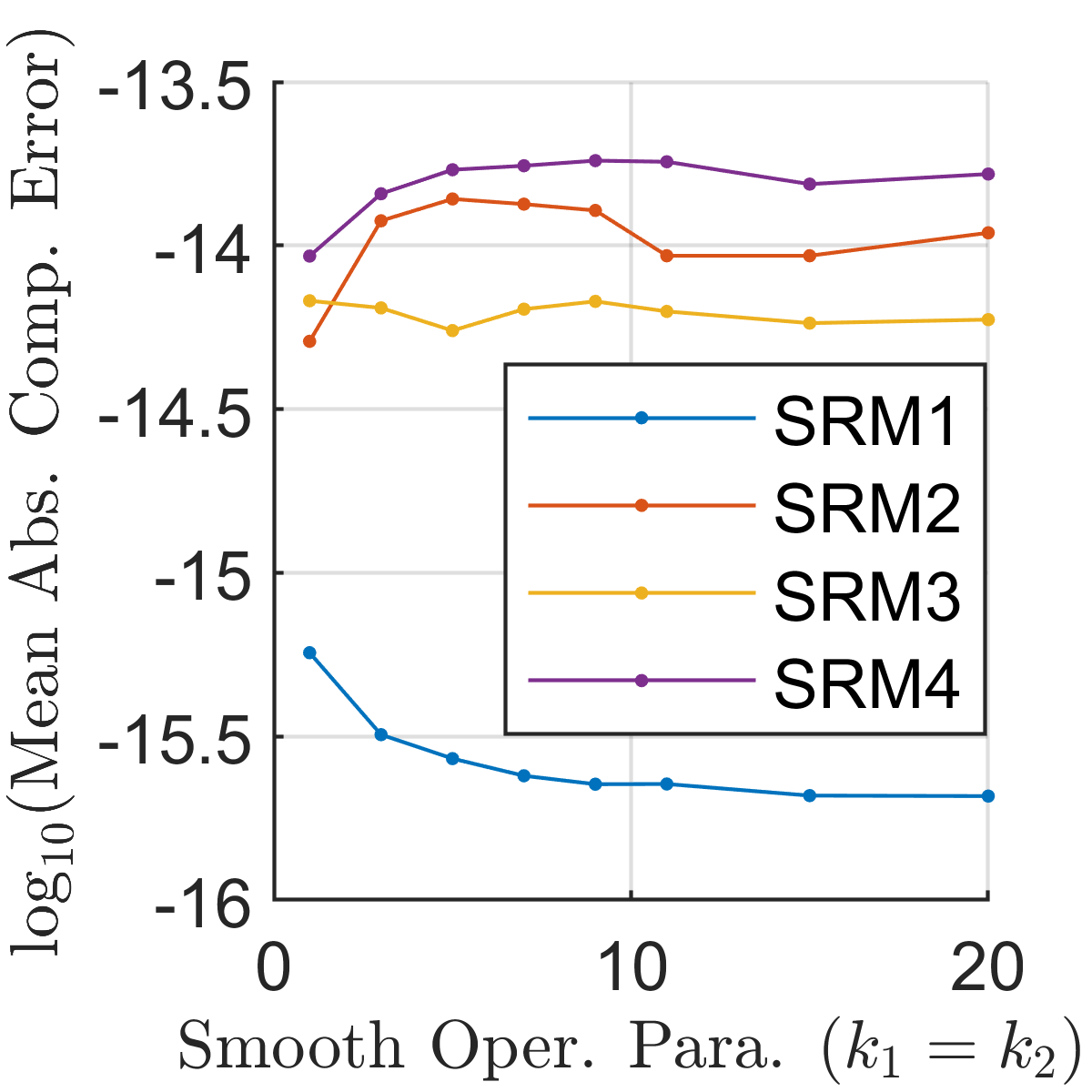}
    \caption{Variation of the  ($\log$) mean absolute component error (between explicit and AdaGrad gradients) with respect to smooth operator parameter values (i.e., $k_1,k_2$) under different SRMs for the SCP3.}
    \label{Fig:GradientResults1}
\end{figure}

\subsection{Effectiveness of different SRMs}


We now compare different control solutions obtained for each SCP shown in Fig. \ref{Fig:ProblemConfigurations} under different SRMs discussed in Section \ref{Sec:SmoothRobustnessMeasures}. Here we also provide the details of respective approximation error bounds revealed by the STL error semantics proposed in Section \ref{Sec:ApproximationErrors}. As mentioned before, in each scenario, similar to \cite{Gilpin2021}, we obtain the control solution by optimizing the smooth cost function $\tilde{J}(u)$ \eqref{Eq:SmoothCostFunction} using the SciPy's SQP method \cite{Virtanen2020}. However, one key difference compared to \cite{Gilpin2021} is that we now provide explicitly computed gradients (i.e., $\frac{\partial \tilde{J}(u)}{\partial u}$) to this SQP solver. 

For the purpose of comparing different control solutions obtained for a SCP, we can directly use the performance metrics: 1) the approximation error bound width $U_\theta^\varphi(s(u))-L_\theta^\varphi(s(u))$, 2) the control cost $0.01\Vert u \Vert^2$, 3) the non-smooth (actual) robustness measure $\rho^\varphi(s(u))$ and 4) the non-smooth (total) cost function $J(u)$. These performance metrics, along with a few other metrics (that cannot be used for such comparisons, e.g., the smooth robustness measure $\tilde{\rho}(s(u))$) that were observed are summarized in Tab. \ref{Tab:SRMPerformanceResults1}. It is important to point out that all the values provided in Tab. \ref{Tab:SRMPerformanceResults1} are average values computed over 50 realizations, where in each realization, the SQP solver was initialized with a randomly selected control solution.

According to the simulation results reported in Tab. \ref{Tab:SRMPerformanceResults1}, in each SCP, the control solution found using the SRM3 provides the best total cost function value and the best actual robustness value, while the control solution found using the SRM1 provides the best control cost and the best approximation error bound width. This behavior implies that: 1) using an SRM with a moderate-sized approximation error bound (like the SRM3, as opposed to the SRM1) can lead to control solutions with better actual robustness measures, and 2) different SRMs can favor different aspects of the composite cost function (e.g., SRM1 favors reducing the control cost while the SRM2 favors improving robustness measure). Moreover, based on the performance metrics considered in Tab. \ref{Tab:SRMPerformanceResults1}, notice that both SRM2 and SRM4 do not show noticeable results compared to SRM1 or SRM3. In fact, for the SCP4, on average, SRM2 and SRM4 fail to find a feasible solution (as $\rho^\varphi<0$). 

To further investigate the properties of different SRMs, we executed the same set of experiments (that generated the results reported in Tab. \ref{Tab:SRMPerformanceResults1}) with different smooth operator parameter values: $k_1,k_2\in\{1,3,5,7,9\}$. The observed results (omitting some indecisive cases for simplicity) are summarized in Tab. \ref{Tab:SRMPerformanceResults2}. Similar to before, these observations show that across different $k_1, k_2$ values and SCPs, the SRM1 provides the best control cost and the best approximation error bound width (if it finds a feasible solution). Also, in terms of the actual robustness measure and the total cost function value, the SRM3 performs better than the SRM1 in 65\% of the cases considered. Furthermore, the results in Tab. \ref{Tab:SRMPerformanceResults2} show that: 1) the overall best robustness measure value achieved for each SCP has been achieved when using the SRM3, 2) the SRM3 is more effective (compared to the SRM1) with low $k_1, k_2$ values, 3) the SRM3 finds a feasible solution irrespective of the $k_1, k_2$ values used, 4) the SRM1 is more effective (compared to the SRM3) with high $k_1, k_2$ values, and, 5) the SRM1 can fail to find a feasible solution when the used $k_1, k_2$ values are small. 

Naturally, the aforementioned conflicting properties among different SRMs encourage one to jointly use two (or more) SRMs to achieve better performance metrics. For example, a control solution obtained using the SRM3 can be applied to initialize a subsequent stage where the SRM1 will be used. In fact, when this exact strategy was applied to the considered four SCPs (with $k_1=k_2=3$, as in Tab. \ref{Tab:SRMPerformanceResults1}), the percentage improvements achieved in terms of the average actual robustness measure were: 0.17\%, 1.80\%, 45.9\% and 0.83\%, respectively.

Another example approach (to jointly use two SRMs) is to constantly switch between using the SRM2 and the SRM3 while following a gradient ascent scheme that optimizes the smooth cost function \eqref{Eq:SmoothCostFunction}. Recall that the SRM2 and the SRM3 are essentially under- and over-approximations to the actual robustness measure. Hence, as illustrated in Fig. \ref{Fig:Switching}, with a properly selected switching strategy, this approach will guarantee an approximation error bound width $(\tilde{\rho}^\varphi_3 - \tilde{\rho}^\varphi_2)$, where $\tilde{\rho}^\varphi_3$ and $\tilde{\rho}^\varphi_2$ are the converged SRM2 and SRM3 values, respectively. For example, for the four considered SCPs (with $k_1=k_2=3$), if the reported SRM2 and SRM3 values in Tab. \ref{Tab:SRMPerformanceResults1} were assumed to be the converged respective SRM2 and SRM3 values under a switching-gradient-ascent approach (as in Fig. \ref{Fig:Switching}), the percentage improvements in the approximation error bound widths would be: 63.3\%, 71.7\%, 77.7\% and 44.1\%, respectively.

\begin{figure}[!t]
    \centering
    \includegraphics[width=2in]{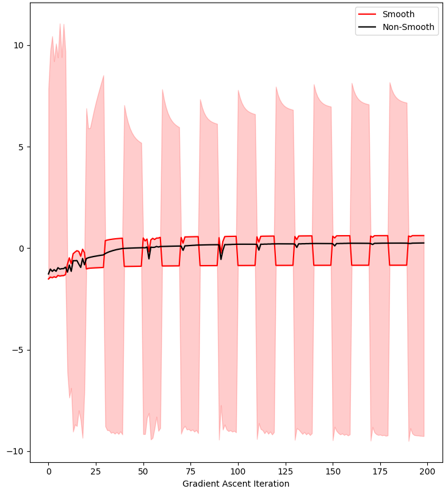}
    \caption{An example case: Switching between SRM2 and SRM3, while following a gradient ascent scheme to optimize the smooth cost function  \eqref{Eq:SmoothCostFunction}. Light red colored regions indicate the error bounds.}
    \label{Fig:Switching}
\end{figure}


\begin{table}[!h]
\caption{A performance comparison of different control solutions observed under different SRMs. Each SRM was defined using the smooth operator parameters $k_1=k_2=3$. All the reported values are average values computed over 50 realizations.}
\label{Tab:SRMPerformanceResults1}
\centering
\resizebox{\columnwidth}{!}{
\begin{tabular}{|c|c|r|r|r|r|r|r|r|r|}
\Xhline{2\arrayrulewidth}
\multirow{2}{*}{} &
  \multirow{2}{*}{} &
  \multirow{2}{*}{} &
  \multicolumn{2}{c|}{Smooth} &
  \multicolumn{3}{c|}{Error Bound Information} &
  \multicolumn{2}{c|}{Non-Smooth} \\ \cline{4-10} 
 SCP &
 SRM &
 \begin{tabular}{c}Cont.\\Cost\end{tabular} &
  \multicolumn{1}{c|}{\begin{tabular}{c}Cost\\$\tilde{J}$\end{tabular}} &
  \multicolumn{1}{c|}{\begin{tabular}{c}Rob.\\$\tilde{\rho}^\varphi$\end{tabular}} &
  \multicolumn{1}{c|}{\begin{tabular}{c}$L_\theta^\varphi$\end{tabular}} &
  \multicolumn{1}{c|}{\begin{tabular}{c}$U_\theta^\varphi$\end{tabular}} &
  \multicolumn{1}{c|}{Width} &
  \multicolumn{1}{c|}{\begin{tabular}{c}Rob.\\$\rho^\varphi$\end{tabular}} &
  \multicolumn{1}{c|}{\begin{tabular}{c}Cost\\$J$\end{tabular}}
  \\ \Xhline{2\arrayrulewidth}
\multirow{4}{*}{SCP1} & SRM1 & \textbf{0.072} & -0.810 & -0.739 & -1.025  & 1.853 & \textbf{2.878} & 0.182          & 0.110          \\ \cline{2-10} 
                       & SRM2 & 0.074          & -0.847 & -0.772 & -0.010  & 4.245 & 4.255          & 0.338          & 0.264          \\ \cline{2-10} 
                       & SRM3 & 0.079          & 0.620  & 0.698  & -3.982  & 0.010 & 3.992          & \textbf{0.401} & \textbf{0.323} \\ \cline{2-10} 
                       & SRM4 & 0.096          & 0.395  & 0.491  & -3.373  & 4.419 & 7.792          & 0.152          & 0.057          \\ \Xhline{2\arrayrulewidth}
\multirow{4}{*}{SCP2} & SRM1 & \textbf{0.093} & -0.892 & -0.799 & -1.025  & 1.853 & \textbf{2.878} & 0.264          & 0.170          \\ \cline{2-10} 
                       & SRM2 & 0.100          & -0.942 & -0.843 & -0.010  & 5.184 & 5.194          & 0.150          & 0.050          \\ \cline{2-10} 
                       & SRM3 & 0.102          & 0.528  & 0.630  & -10.037 & 0.010 & 10.047         & \textbf{0.321} & \textbf{0.220} \\ \cline{2-10} 
                       & SRM4 & 0.116          & 0.352  & 0.468  & -9.543  & 4.110 & 13.653         & 0.155          & 0.040          \\ \Xhline{2\arrayrulewidth}
\multirow{4}{*}{SCP3} & SRM1 & \textbf{0.214} & -0.416 & -0.202 & -0.703  & 1.853 & \textbf{2.556} & 0.524          & 0.310          \\ \cline{2-10} 
                       & SRM2 & 0.245          & -0.522 & -0.276 & -0.010  & 6.266 & 6.276          & 0.523          & 0.278          \\ \cline{2-10} 
                       & SRM3 & 0.236          & 0.856  & 1.091  & -6.176  & 0.010 & 6.186          & \textbf{0.620} & \textbf{0.385} \\ \cline{2-10} 
                       & SRM4 & 0.276          & 0.625  & 0.901  & -5.359  & 3.823 & 9.182          & 0.463          & 0.187          \\ \Xhline{2\arrayrulewidth}
\multirow{4}{*}{SCP4} & SRM1 & \textbf{0.231} & -0.673 & -0.442 & -0.703  & 1.988 & \textbf{2.691} & 0.410          & 0.179          \\ \cline{2-10} 
                       & SRM2 & \sout{0.167}          & \sout{-1.896} & \sout{-1.729} & \sout{-0.010}  & \sout{4.409} & \sout{4.419}          & \textcolor{red}{-1.345}         & \sout{-1.512}         \\ \cline{2-10} 
                       & SRM3 & 0.242          & 0.574  & 0.816  & -7.107  & 0.010 & 7.117          & \textbf{0.506} & \textbf{0.264} \\ \cline{2-10} 
                       & SRM4 & \sout{0.147}          & \sout{-1.450} & \sout{-1.303} & \sout{-9.295}  & \sout{4.526} & \sout{13.822}         & \textcolor{red}{-1.535}         & \sout{-1.681}         \\ \Xhline{2\arrayrulewidth}
\end{tabular}}
\end{table}

\begin{table*}[!t]
\caption{A performance comparison of different control solutions observed under SRM1 and SRM3 with varying smooth operator parameter values $k_1,k_2\in\{1,3,5,7,9\}$. All the reported values are average values computed over 50 realizations.}
\label{Tab:SRMPerformanceResults2}
\resizebox{\textwidth}{!}{
\begin{tabular}{|c|c|r|r|r|r|r|r|r|r|r|r|r|r|r|r|r|r|r|r|r|r|}
\Xhline{2\arrayrulewidth}
\multirow{3}{*}{SCP} &
  \multirow{3}{*}{SRM} &
  \multicolumn{5}{c|}{Error Bound width: $U^\varphi_\theta(s(u)) - L^\varphi_\theta(s(u))$} &
  \multicolumn{5}{c|}{Control Cost: $0.01 \Vert u \Vert^2$} &
  \multicolumn{5}{c|}{Robustness Measure: $\rho^\varphi(s(u))$} &
  \multicolumn{5}{c|}{Cost: $J(u) = \rho^\varphi(s(u)) - 0.01 \Vert u \Vert^2$} \\ \cline{3-22} 
 &
   &
  \multicolumn{5}{c|}{k\_1=k\_2} &
  \multicolumn{5}{c|}{k\_1=k\_2} &
  \multicolumn{5}{c|}{k\_1 = k\_2} &
  \multicolumn{5}{c|}{k\_1 = k\_2} \\ \cline{3-22} 
 &
   &
  \multicolumn{1}{c|}{1} &
  \multicolumn{1}{c|}{3} &
  \multicolumn{1}{c|}{5} &
  \multicolumn{1}{c|}{7} &
  \multicolumn{1}{c|}{9} &
  \multicolumn{1}{c|}{1} &
  \multicolumn{1}{c|}{3} &
  \multicolumn{1}{c|}{5} &
  \multicolumn{1}{c|}{7} &
  \multicolumn{1}{c|}{9} &
  \multicolumn{1}{c|}{1} &
  \multicolumn{1}{c|}{3} &
  \multicolumn{1}{c|}{5} &
  \multicolumn{1}{c|}{7} &
  \multicolumn{1}{c|}{9} &
  \multicolumn{1}{c|}{1} &
  \multicolumn{1}{c|}{3} &
  \multicolumn{1}{c|}{5} &
  \multicolumn{1}{c|}{7} &
  \multicolumn{1}{c|}{9} \\ \Xhline{2\arrayrulewidth}
\multirow{2}{*}{SCP1} &
  SRM1 &
  \sout{8.594} &
  \textbf{2.878} &
  \textbf{1.735} &
  \textbf{1.245} &
  {\ul \textbf{0.973}} &
  \sout{0.068} &
  {\ul \textbf{0.072}} &
  \textbf{0.072} &
  \textbf{0.073} &
  \textbf{0.073} &
  \textcolor{red}{-1.213} &
  0.182 &
  0.349 &
  0.400 &
  \textbf{0.416} &
  \sout{-1.281} &
  0.110 &
  0.277 &
  \textbf{0.327} &
  {\ul \textbf{0.343}} \\ \cline{2-22} 
 &
  SRM3 &
  \textbf{8.635} &
  3.992 &
  3.420 &
  2.766 &
  2.946 &
  \textbf{0.080} &
  0.079 &
  0.079 &
  0.079 &
  0.078 &
  \textbf{0.192} &
  \textbf{0.401} &
  {\ul \textbf{0.416}} &
  \textbf{0.401} &
  0.399 &
  \textbf{0.112} &
  \textbf{0.323} &
  \textbf{0.338} &
  0.323 &
  0.320 \\ \Xhline{2\arrayrulewidth}
\multirow{2}{*}{SCP2} &
  SRM1 &
  \sout{8.594} &
  \textbf{2.878} &
  \textbf{1.735} &
  \textbf{1.245} &
  {\ul \textbf{0.973}} &
  \sout{0.097} &
  {\ul \textbf{0.093}} &
  \textbf{0.094} &
  \textbf{0.094} &
  \textbf{0.093} &
  \textcolor{red}{-0.896} &
  0.264 &
  0.335 &
  \textbf{0.360} &
  \textbf{0.372} &
  \sout{-0.993} &
  0.170 &
  0.241 &
  \textbf{0.266} &
  {\ul \textbf{0.278}} \\ \cline{2-22} 
 &
  SRM3 &
  \textbf{11.648} &
  10.047 &
  9.456 &
  8.982 &
  8.741 &
  \textbf{0.102} &
  0.102 &
  0.098 &
  0.100 &
  0.098 &
  \textbf{0.185} &
  \textbf{0.321} &
  {\ul \textbf{0.372}} &
  0.360 &
  0.346 &
  \textbf{0.083} &
  \textbf{0.220} &
  \textbf{0.274} &
  0.260 &
  0.248 \\ \Xhline{2\arrayrulewidth}
\multirow{2}{*}{SCP3} &
  SRM1 &
  \textbf{7.629} &
  \textbf{2.556} &
  \textbf{1.542} &
  \textbf{1.107} &
  {\ul \textbf{0.865}} &
  {\ul \textbf{0.118}} &
  \textbf{0.214} &
  \textbf{0.225} &
  \textbf{0.223} &
  \textbf{0.219} &
  0.187 &
  0.524 &
  \textbf{0.602} &
  \textbf{0.596} &
  \textbf{0.590} &
  0.069 &
  0.310 &
  \textbf{0.376} &
  \textbf{0.374} &
  \textbf{0.371} \\ \cline{2-22} 
 &
  SRM3 &
  7.251 &
  6.186 &
  5.498 &
  5.105 &
  5.179 &
  0.165 &
  0.236 &
  0.248 &
  0.235 &
  0.226 &
  \textbf{0.542} &
  {\ul \textbf{0.620}} &
  0.568 &
  0.590 &
  0.574 &
  \textbf{0.376} &
  {\ul \textbf{0.385}} &
  0.320 &
  0.355 &
  0.348 \\ \Xhline{2\arrayrulewidth}
\multirow{2}{*}{SCP4} &
  SRM1 &
  \textbf{8.034} &
  \textbf{2.691} &
  \textbf{1.623} &
  \textbf{1.165} &
  {\ul \textbf{0.910}} &
  \textbf{0.227} &
  \textbf{0.231} &
  \textbf{0.229} &
  0.222 &
  \textbf{0.219} &
  0.061 &
  0.410 &
  \textbf{0.506} &
  0.512 &
  0.524 &
  -0.166 &
  0.179 &
  \textbf{0.276} &
  0.290 &
  0.304 \\ \cline{2-22} 
 &
  SRM3 &
  8.866 &
  7.117 &
  6.451 &
  6.348 &
  6.021 &
  0.231 &
  0.242 &
  0.235 &
  {\ul \textbf{0.217}} &
  0.220 &
  \textbf{0.151} &
  \textbf{0.506} &
  0.494 &
  \textbf{0.532} &
  {\ul \textbf{0.538}} &
  \textbf{-0.080} &
  \textbf{0.264} &
  0.258 &
  \textbf{0.315} &
  {\ul \textbf{0.317}} \\ \Xhline{2\arrayrulewidth}
\end{tabular}}
\end{table*}

\subsection{The use of approximation error bounds}

We finally highlight three possible interesting uses of approximation error bounds determined by the STL error semantics proposed in Section \ref{Sec:ApproximationErrors}. Before getting into the details, recall the notation that we used to represent an approximation error bound: $[L_\theta^\varphi,U_\theta^\varphi] \ni \tilde{e}^\varphi(s(u)) \triangleq (\rho^\varphi(s(u)) - \tilde{\rho}^\varphi(s(u))\})$ and the fact that we optimize a SRM $\tilde{\rho}^\varphi(s(u))$ using a gradient-based optimization process. 

The first use of knowing the error bound $[L_\theta^\varphi,U_\theta^\varphi]$ is that it allows one to terminate the said optimization process prematurely if a given minimum level of actual robustness $\rho^\varphi(s(u))$ is guaranteed to be achieved. For example, if we need to ensure $\rho^\varphi(s(u))>L$ for some known $L\in\R$ value, we can use the simple (also efficient and smooth) termination condition: $\tilde{\rho}^\varphi(s(u))\geq L-L_\theta^\varphi$ in the gradient-based optimization process. As a numerical example, consider the (average) case for the SCP1 with the SRM2 reported in Tab. \ref{Tab:SRMPerformanceResults1} and assume we need to ensure $\rho^\varphi(s(u))>-1$. By simply looking at the achieved $\tilde{\rho}^\varphi(s(u))$ and $L_\theta^\varphi$ values, one can directly conclude that $\rho^\varphi(s(u))>-1$ as $\tilde{\rho}^\varphi(s(u)) = -0.772 \geq L-L_\theta^\varphi = -1-(-0.01) = -0.99$.

The next use of error bounds is that it allows one to see how uncertainties in the mission space translate into uncertainties in the actual robustness measure. For example, in Tab. \ref{Tab:SRMPerformanceResults1}, under any SCP, notice that $L^\varphi_\theta = -0.01$ in SRM2 and $U^\varphi_\theta = 0.01$ in SRM3 (even though, they both should be zero according to Lms. \ref{Lm:SRM2ApproxError} and \ref{Lm:SRM3ApproxError}, respectively). This is because, in our experiments, we have assumed that each predicate value (i.e., $\mu^\pi(s_t)$) can be deviated by $\pm0.01$ from its true value due to some source of noise. In other words, we set $L^\pi(s_t) = L^\pi = -0.01$ and $U^\pi(s_t) = U^\pi = 0.01$ for all $s_t$ and $\pi$, in the experiments (see also \eqref{Eq:ModifiedSmoothSTLRobustSemantic}). In applications, since these $L^\pi$ and $U^\pi$ values can take different values depending on the predicate $\pi$, it is far from obvious that how such terms will determine an accuracy error bound: $[L^\varphi,U^\varphi] \ni \hat{e}^\varphi(s(u)) \triangleq (\rho^\varphi(s(u)) - \hat{\rho}^\varphi(s(u)))$, where $\hat{\rho}^\varphi(s(u))$ is the computed/estimated actual robustness measure and $\rho^\varphi(s(u))$ is the true actual robustness measure. However, note that we now can address this challenge by simply using the supplementary STL error semantics proposed in Def. \ref{Def:SuppSmoothSTLRobustSemantics} with $L_{k_1,m}^{\min} = U_{k_1,m}^{\min} = L_{k_2,m}^{\min} = U_{k_2,m}^{\min} = 0$ to compute the said accuracy error bound $[L^\varphi,U^\varphi]$.

Finally, we elaborate on the main intended use of the proposed error bound analysis, i.e., to select (tune) the smooth operator parameters $\theta = \{k_1,k_2,...\}$. Intuitively, for this purpose, one can solve an optimization problem of the form: 
\begin{equation}\label{Eq:ParameterTuning}
    \theta^* = \underset{\theta > 0}{\arg\min}\ \ (U^\varphi_\theta - L^\varphi_\theta) + \alpha \theta^T\theta,
\end{equation}
where $\alpha$ is a scaling factor and the second term in the RHS is a regulatory term that ensures the parameters in $\theta$ will not grow indefinitely. Note that this latter term is required because the error bound width $(U^\varphi_\theta - L^\varphi_\theta)$ term in \eqref{Eq:ParameterTuning} converges to zero as the parameters in $\theta$ grow (irrespective of the used SRM). Clearly, with respect to the optimization process of the cost function \eqref{Eq:SmoothCostFunction}, this parameter tuning stage \eqref{Eq:ParameterTuning} can be executed either off-line or on-line. A situation where off-line parameter tuning has lead to an improved solution is shown in Fig. \ref{Fig:OfflineTuning}. Further, Fig. \ref{Fig:OnlineTuning} shows a case where on-line parameter tuning has been used - while the cost function  \eqref{Eq:SmoothCostFunction} is being optimized using a gradient ascent process. From these motivating results, the importance of the proposed error bound analysis is evident. The future work is directed towards studying efficient and accurate ways to solve this parameter tuning problem \eqref{Eq:ParameterTuning} and investigating effective ways to combine it with the main optimization problem where the cost function \eqref{Eq:SmoothCostFunction} is optimized.

\begin{figure}[!t]
    \centering
    \begin{subfigure}[t]{0.48\columnwidth}
        \centering
        \captionsetup{justification=centering}
        \includegraphics[width=\textwidth]{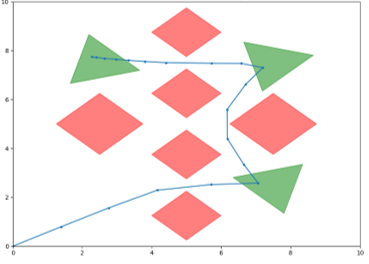}
        \caption{$\rho^\varphi(s(u^*)) = 0.423$, $[L_\theta^\varphi, U_\theta^\varphi] \equiv [-0.7031, 1.9883]$, $\theta = [3.000, 3.000, 3.000,\ldots]$}
    \end{subfigure}
    \hfill
    \begin{subfigure}[t]{0.48\columnwidth}
        \centering
        \captionsetup{justification=centering}
        \includegraphics[width=\textwidth]{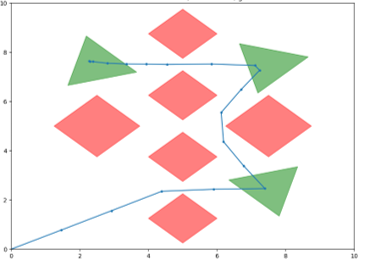}
        \caption{$\rho^\varphi(s(u^*)) = 0.493$,  $[L_\theta^\varphi, U_\theta^\varphi] \equiv [-0.4999, 0.9457]$, $\theta = [2.715, 4.038, 2.776,\ldots]$}
    \end{subfigure}
    \caption{Off-line smooth parameter tuning (of the SRM1): Generated results from the SQP solver: (a) with and (b) without off-line parameter tuning.}
    \label{Fig:offlineTuning}
\end{figure}

\begin{figure}[!t]
    \centering
    \begin{subfigure}[t]{0.48\columnwidth}
        \centering
        \captionsetup{justification=centering}
        \includegraphics[width=\textwidth]{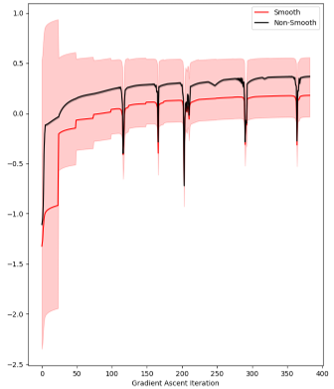}
        \caption{The evolution of $\tilde{\rho}^\varphi, \rho^\varphi$ and $[L_\theta^\varphi, U_\theta^\varphi]$ over the gradient ascent steps.}
    \end{subfigure}
    \hfill
    \begin{subfigure}[t]{0.48\columnwidth}
        \centering
        \captionsetup{justification=centering}
        \includegraphics[width=\textwidth]{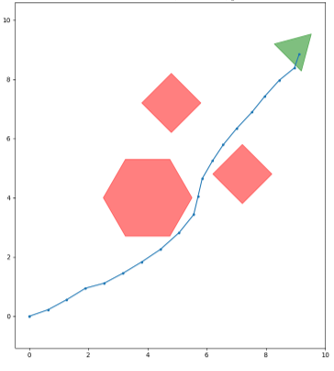}
        \caption{The converged result: $\rho^\varphi(s(u^*)) = 0.367$ (39.1\% better than the average: 0.264).}
    \end{subfigure}
    \caption{On-line smooth operator parameter tuning (of the SRM1): The cost function \eqref{Eq:SmoothCostFunction} was optimized using a gradient ascent process while tuning smooth operator parameters $\theta$ as per \eqref{Eq:ParameterTuning} every 25 iterations.}
    \label{Fig:OnlineTuning}
\end{figure}

\section{Conclusion}
\label{Sec:Conclusion}

In this paper, we motivated the use of smooth robustness measures and gradient-based optimization methods to efficiently solve complex symbolic control problems where the specifications are given using Signal Temporal Logic. First, two existing and two new smooth robustness measures were extensively discussed while establishing their fundamental properties. Then, to characterize the approximation errors associated with different SRMs, we proposed a set of semantic rules that we named \emph{STL error semantics}. This error analysis, among its other uses, enables one to assess different SRMs and sensibly select SRM parameters. Finally, to efficiently and accurately compute the gradients of different SRMs, we proposed another set of semantic rules that we named \emph{STL gradient semantics}. Simulation results are provided to highlight the improvements achieved due to these contributions. Ongoing work aims to determine effective ways to use multiple SRMs in gradient-based optimization schemes simultaneously.

\bibliographystyle{IEEEtran}
\bibliography{References}

\begin{thebibliography}{10}
\providecommand{\url}[1]{#1}
\csname url@samestyle\endcsname
\providecommand{\newblock}{\relax}
\providecommand{\bibinfo}[2]{#2}
\providecommand{\BIBentrySTDinterwordspacing}{\spaceskip=0pt\relax}
\providecommand{\BIBentryALTinterwordstretchfactor}{4}
\providecommand{\BIBentryALTinterwordspacing}{\spaceskip=\fontdimen2\font plus
\BIBentryALTinterwordstretchfactor\fontdimen3\font minus
  \fontdimen4\font\relax}
\providecommand{\BIBforeignlanguage}[2]{{%
\expandafter\ifx\csname l@#1\endcsname\relax
\typeout{** WARNING: IEEEtran.bst: No hyphenation pattern has been}%
\typeout{** loaded for the language `#1'. Using the pattern for}%
\typeout{** the default language instead.}%
\else
\language=\csname l@#1\endcsname
\fi
#2}}
\providecommand{\BIBdecl}{\relax}
\BIBdecl

\bibitem{Pnueli1977}
A.~Pnueli, ``{The Temporal Logic of Programs},'' in \emph{Proc of Annual IEEE
  Symposium on Foundations of Computer Science}, vol. 1977-October, 1977, pp.
  46--57.

\bibitem{Clarke1982}
E.~M. Clarke and E.~A. Emerson, ``{Design and Synthesis of Synchronization
  Skeletons Using Branching Time Temporal Logic},'' in \emph{Logics of
  Programs}, 1982, pp. 52--71.

\bibitem{Schwartz1983}
R.~L. Schwartz, P.~M. Melliar-Smith, and F.~H. Vogt, ``{An Interval Logic for
  Higher-Level Temporal Reasoning},'' in \emph{Proc. of 2nd Annual ACM
  Symposium on Principles of Distributed Computing}.\hskip 1em plus 0.5em minus
  0.4em\relax New York, New York, USA: ACM Press, 1983, pp. 173--186.

\bibitem{Koymans1990}
R.~Koymans, ``{Specifying Real-Time Properties with Metric Temporal Logic},''
  \emph{Real-Time Systems}, vol.~2, no.~4, pp. 255--299, 1990.

\bibitem{Vasile2017}
C.~I. Vasile, D.~Aksaray, and C.~Belta, ``{Time Window Temporal Logic},''
  \emph{Theoretical Computer Science}, vol. 691, pp. 27--54, 2017.

\bibitem{Maler2004}
O.~Maler and D.~Nickovic, ``{Monitoring Temporal Properties of Continuous
  Signals},'' \emph{Lecture Notes in Computer Science}, vol. 3253, pp.
  152--166, 2004.

\bibitem{Mehdipour2019}
N.~Mehdipour, C.~I. Vasile, and C.~Belta, ``{Arithmetic-Geometric Mean
  Robustness for Control From Signal Temporal Logic Specifications},'' in
  \emph{Proc. of American Control Conf.}, vol. 2019-July, 2019, pp. 1690--1695.

\bibitem{Haghighi2019}
I.~Haghighi, N.~Mehdipour, E.~Bartocci, and C.~Belta, ``{Control from Signal
  Temporal Logic Specifications with Smooth Cumulative Quantitative
  Semantics},'' in \emph{Proc. of 58th IEEE Conf. on Decision and Control},
  vol. 2019-Decem, 2019, pp. 4361--4366.

\bibitem{Pant2018}
Y.~V. Pant, H.~Abbas, R.~A. Quaye, and R.~Mangharam, ``{Fly-by-Logic: Control
  of Multi-Drone Fleets with Temporal Logic Objectives},'' in \emph{Proc. of
  9th ACM/IEEE Intl. Conf. on Cyber-Physical Systems}, 2018, pp. 186--197.

\bibitem{Fainekos2012}
G.~E. Fainekos, S.~Sankaranarayanan, K.~Ueda, and H.~Yazarel, ``{Verification
  of Automotive Control Applications Using S-TaLiRo},'' in \emph{Proc. of
  American Control Conf.}, 2012, pp. 3567--3572.

\bibitem{Raman2015}
V.~Raman, A.~Donz{\'{e}}, D.~Sadigh, R.~M. Murray, and S.~A. Seshia,
  ``{Reactive Synthesis from Signal Temporal Logic Specifications},'' in
  \emph{Proc. of 18th Intl. Conf. on Hybrid Systems: Computation and Control},
  2015, pp. 239--248.

\bibitem{Sankaranarayanan2012}
S.~Sankaranarayanan and G.~Fainekos, ``{Simulating Insulin Infusion Pump Risks
  by In-Silico Modeling of the Insulin-Glucose Regulatory System},'' in
  \emph{Proc. of Intl. Conf. on Computational Methods in Systems Biology},
  2012, pp. 322--341.

\bibitem{Mehdipour2019b}
N.~Mehdipour, D.~Briers, I.~Haghighi, C.~M. Glen, M.~L. Kemp, and C.~Belta,
  ``{Spatial-Temporal Pattern Synthesis in a Network of Locally Interacting
  Cells},'' in \emph{Proc. of 58th IEEE Conf. on Decision and Control}, vol.
  2018-Decem, 2019, pp. 3516--3521.

\bibitem{Belta2019}
C.~Belta and S.~Sadraddini, ``{Formal Methods for Control Synthesis: An
  Optimization Perspective},'' \emph{Annual Review of Control, Robotics, and
  Autonomous Systems}, vol.~2, no.~1, pp. 115--140, 2019.

\bibitem{Saha2016}
S.~Saha and A.~A. Julius, ``{An MILP Approach for Real-Time Optimal Controller
  Synthesis with Metric Temporal Logic Specifications},'' in \emph{Proc. of the
  American Control Conf.}, vol. 2016-July, 2016, pp. 1105--1110.

\bibitem{Abbas2012}
H.~Abbas and G.~Fainekos, ``{Convergence Proofs for Simulated Annealing
  Falsification of Safety Properties},'' in \emph{Proc. of 50th Annual Allerton
  Conference on Communication, Control, and Computing}, 2012, pp. 1594--1601.

\bibitem{Lavalle2001}
S.~M. LaValle and J.~{James J. Kuffner}, ``{Randomized Kinodynamic Planning},''
  \emph{The International Journal of Robotics Research}, vol.~20, no.~5, pp.
  378--400, 2001.

\bibitem{Abbas2013}
H.~Abbas and G.~Fainekos, ``{Computing Descent Direction of MTL Robustness for
  Non-Linear Systems},'' in \emph{Proc. of American Control Conf.}, 2013, pp.
  4405--4410.

\bibitem{Abbas2014}
H.~Abbas, A.~Winn, G.~Fainekos, and A.~A. Julius, ``{Functional Gradient
  Descent Method for Metric Temporal Logic Specifications},'' in \emph{Proc. of
  American Control Conference}, 2014, pp. 2312--2317.

\bibitem{Gilpin2021}
Y.~Gilpin, V.~Kurtz, and H.~Lin, ``{A Smooth Robustness Measure of Signal
  Temporal Logic for Symbolic Control},'' \emph{IEEE Control Systems Letters},
  vol.~5, no.~1, pp. 241--246, 2021.

\bibitem{Pant2017}
Y.~V. Pant, H.~Abbas, and R.~Mangharam, ``{Smooth Operator: Control Using the
  Smooth Robustness of Temporal Logic},'' in \emph{Proc. of 1st IEEE Conf. on
  Control Technology and Applications}, vol. 2017-Janua, 2017, pp. 1235--1240.

\bibitem{Lasdon2008}
L.~Lasdon and J.~C. Plummer, ``{Multistart algorithms for seeking
  feasibility},'' \emph{Computers and Operations Research}, vol.~35, no.~5, pp.
  1379--1393, may 2008.

\bibitem{Welikala2019J1}
S.~Welikala and C.~G. Cassandras, ``{Distributed Non-Convex Optimization of
  Multi-Agent Systems Using Boosting Functions to Escape Local Optima},''
  \emph{IEEE Trans. on Automatic Control}, 2020.

\bibitem{Hoos2005}
H.~H. Hoos and T.~Stutzle, ``{Praise for Stochastic Local Search: Foundations
  and Applications},'' in \emph{Stochastic Local Search}.\hskip 1em plus 0.5em
  minus 0.4em\relax Morgan Kaufmann Publishers, 2005, pp. i--ii.

\bibitem{Lange2014}
M.~Lange, D.~Z{\"{u}}hlke, O.~Holz, and T.~Villmann, ``{Applications of
  lp-Norms and Their Smooth Approximations for Gradient Based Learning Vector
  Quantization},'' in \emph{Proc. of European Symposium on Artificial Neural
  Networks}, 2014, pp. 271--276.

\bibitem{Maclaurin2015}
D.~Maclaurin, D.~Duvenaud, and R.~P. Adams, ``{Autograd: Effortless Gradients
  in numpy},'' in \emph{Proc. of ICML AutoML Workshop}, 2015, pp. 1--3.

\bibitem{Virtanen2020}
\BIBentryALTinterwordspacing
P.~Virtanen, R.~Gommers, T.~E. Oliphant, M.~Haberland, T.~Reddy, D.~Cournapeau,
  E.~Burovski, P.~Peterson, W.~Weckesser, J.~Bright, S.~J. van~der Walt,
  M.~Brett, J.~Wilson, K.~J. Millman, N.~Mayorov, A.~R.~J. Nelson, E.~Jones,
  R.~Kern, E.~Larson, C.~J. Carey, I.~Polat, Y.~Feng, E.~W. Moore,
  J.~VanderPlas, D.~Laxalde, J.~Perktold, R.~Cimrman, I.~Henriksen, E.~A.
  Quintero, C.~R. Harris, A.~M. Archibald, A.~H. Ribeiro, F.~Pedregosa, and
  P.~van Mulbregt, ``{SciPy 1.0: Fundamental Algorithms for Scientific
  Computing in Python},'' \emph{Nature Methods}, vol.~17, no.~3, pp. 261--272,
  2020. [Online]. Available:
  \url{https://www.nature.com/articles/s41592-019-0686-2}
\BIBentrySTDinterwordspacing

\end{thebibliography}


\end{document}